\documentclass[conference]{IEEEtran} \newcommand{\version}{long} 
\usepackage{amsmath}
\usepackage{amsthm}
\usepackage{amssymb}
\usepackage{graphicx}
\usepackage{subfigure}
\usepackage{color}
\usepackage{ifthen}
\usepackage{balance}
\usepackage{cite}
\newtheorem{definition}{Definition}

\newtheorem{corollary}{Corollary}
\newtheorem{theorem}{Theorem}
\newcommand{\cond}[2]{\ifthenelse{\equal{\version}{#1}}{#2}{}}
\IEEEoverridecommandlockouts
\begin{document}
\title{Non-Asymptotic Delay Bounds for $(k,l)$ Fork-Join Systems and Multi-Stage Fork-Join Networks}
\author{\IEEEauthorblockN{Markus Fidler\thanks{This work was supported
in part by the European Research Council (ERC) under Starting Grant UnIQue (StG 306644).}}
\IEEEauthorblockA{Institute of Communications Technology\\Leibniz Universit\"{a}t Hannover}
\and
\IEEEauthorblockN{Yuming Jiang}
\IEEEauthorblockA{Department of Telematics\\NTNU Trondheim}}
\maketitle
\cond{long}{
\thispagestyle{plain}
\pagestyle{plain}
}
\begin{abstract}
Parallel systems have received increasing attention with numerous recent applications such as fork-join systems, load-balancing, and $l$-out-of-$k$ redundancy. Common to these systems is a join or resequencing stage, where tasks that have finished service may have to wait for the completion of other tasks so that they leave the system in a predefined order. These synchronization constraints make the analysis of parallel systems challenging and few explicit results are known. In this work, we model parallel systems using a max-plus approach that enables us to derive statistical bounds of waiting and sojourn times. Taking advantage of max-plus system theory, we also show end-to-end delay bounds for multi-stage fork-join networks. We contribute solutions for basic G$\mid$G$\mid$1 fork-join systems, parallel systems with load-balancing, as well as general $(k,l)$ fork-join systems with redundancy. Our results provide insights into the respective advantages of $l$-out-of-$k$ redundancy vs. load-balancing.
\end{abstract}
%
%
\section{Introduction}
\label{sec:introduction}
Fork-join systems are an essential model of parallel data processing, such as Hadoop MapReduce~\cite{tan:mapreduce}, where jobs are divided into $k$ tasks (fork) that are processed in parallel by $k$ servers. Once all tasks of a job are completed, the results are combined (join) and the job leaves the system. Multi-stage fork-join networks comprise several fork-join systems in tandem, where all tasks of a job have to be completed at the current stage before the job is handed over to the next stage. The difficulty in analyzing such systems is due to a) the statistical dependence of the workload of the parallel servers that is due to the common arrival process~\cite{baccelli:forkjoin,kemper:forkjoin}, and b) the synchronization that is enforced by the join operation~\cite{baccelli:forkjoin}.

Significant research has been performed to analyze the performance of fork-join systems. Exact results are, however, known for few specific systems only, such as for two parallel M$\mid$M$\mid$1 queues~\cite{flatto:forkjoin, nelson:forkjoin}. For more complex systems, approximation techniques, e.g.,~\cite{nelson:forkjoin, lebrecht:forkjoin, ko:forkjoin, tan:forkjoin, varma:forkjoin, varki:forkjoin, kemper:forkjoin, alomari:forkjoin}, and bounds, using stochastic orderings~\cite{baccelli:forkjoin}, martingales~\cite{rizk:forkjoin}, or stochastical burstiness constraints~\cite{kesidis:forkjoin}, have been explored. Given the difficulties posed by single-stage fork-join systems, few works consider multi-stage networks. A notable exception is~\cite{varki:forkjoin} where an approximation for closed fork-join networks is developed.

Related synchronization problems occur also in case of load balancing using parallel servers and in case of multi-path routing of packet data streams~\cite{han:resequencing} using multi-path protocols~\cite{rizk:forkjoin}. The tail behavior of delays in multi-path routing is investigated in~\cite{han:resequencing} as well as in~\cite{xia:resequencing,gao:resequencing} where large deviation results of resequencing delays for parallel M$\mid$M$\mid$1 servers are derived.

A further synchronization constraint applies in split-merge systems, that are a variant of fork-join systems, where all tasks of a job have to start execution at the same time. In contrast, in a fork-join system, the start time of tasks is not synchronized. Split-merge systems are solvable to some extend as they can be expressed as a single server queue, where the service process is governed by the maximal service time~\cite{harrison:splitmerge, lebrecht:forkjoin, rizk:forkjoin, joshi:knforkjoin}.

A generalization of fork-join systems are $(k,l)$ fork-join systems, where a job is finished once any $l$ out of its $k$ tasks are completed. The model enables analyzing systems with redundant servers and suitable coding. A first queueing model of maximum distance separable (MDS) codes is presented in~\cite{shah:mdsqueue} and the download of coded content from distributed storage systems is modelled as a $(k,l)$ fork-join system in~\cite{joshi:knforkjoin}. Generally, $(k,l)$ fork-join systems are governed by the $l$th order statistic. The authors of~\cite{joshi:knforkjoin} compute bounds of the mean response time using as approximation a $(k,l)$ split-merge model with Poisson arrivals, that is solved by the Pollaczek-Khinchin formula.

Most closely related to this work are the two recent papers~\cite{kesidis:forkjoin, rizk:forkjoin} that employ similar methods. The work~\cite{kesidis:forkjoin} considers single-stage fork-join systems with load balancing, general arrivals of the type defined in~\cite{yin:generalizedstochasticallyboundedburstiness}, and deterministic service. A service curve characterization of fork-join systems is provided and first statistical delay bounds are presented. The paper~\cite{rizk:forkjoin} contributes delay bounds for single-stage fork-join systems with renewal as well as Markov modulated inter-arrival times and independent and identically distributed (iid) service times. The authors prove that delays for fork-join systems grow as $\mathcal{O}(\ln k)$ for $k$ parallel servers as also found in~\cite{baccelli:forkjoin}. Split-merge systems exhibit an inferior performance, where the stability region of $k$ parallel M$\mid$M$\mid$1 servers is shown to decrease with $\ln k$. The work also includes a first application to multi-path routing, assuming a simple window-based protocol that operates on batches of packets. Response times are computed for entire batches and the authors conclude that multi-path routing is beneficial in case of only two parallel paths and moderate to high utilization. Otherwise, resequencing delays are found to dominate.

While~\cite{rizk:forkjoin} focuses on split-merge vs. fork-join systems with iid service times, we consider also the case of non-iid service, where we are able to generalize important results, such as the growth of delays in $\mathcal{O}(\ln k)$ for fork-join systems with $k$ parallel servers. Further, we investigate advanced fork-join systems beyond~\cite{kesidis:forkjoin, rizk:forkjoin}, including $(k,l)$-fork join systems and multi-stage fork-join networks, where we present a scaling of end-to-end delay bounds in $\mathcal{O}(h\ln(hk))$ for $h$ stages. Considering heterogeneous servers and deterministic as well as random thinning of arrival processes, we give insights into load-balancing. We draw essential conclusions on the advantages of load-balancing and $l$-out-of-$k$ redundancy.

The remainder of this paper is structured as follows. In Sec.~\ref{sec:forkjoin}, we phrase basic models of G$\mid$G$\mid$1 as well as GI$\mid$GI$\mid$1 fork-join systems in max-plus system theory and show an application to load balancing. We extend the model to parallel servers with thinning and resequencing, as it applies, e.g., in multi-path routing, in Sec.~\ref{sec:thinning}. In Sec.~\ref{sec:advanced}, we consider advanced $(k,l)$ fork-join systems and multi-stage fork-join networks. Sec.~\ref{sec:conclusions} presents brief conclusions.
%
%
\section{Basic Fork-Join Systems}
\label{sec:forkjoin}
In this section, we derive a set of results for basic fork-join systems. \cond{long}{Compared to~\cite{rizk:forkjoin}, we define a general queueing model in max-plus system theory~\cite{baccelli:synchronizationlinearity, chang:performanceguarantees, jiang:maxplus, jiang:onecoin, luebben:availbw2} that is a branch of the deterministic~\cite{cruz:networkdelaycalculus, chang:performanceguarantees, leboudec:networkcalculus}, respectively, stochastic network calculus~\cite{chang:performanceguarantees, burchard:endtoendstatisticalcalculus, ciucu:networkservicecurvescaling2, fidler:momentcalculus, jiang:stochasticnetworkcalculus, fidler:netcalcsurvey, ciucu:goodvalue, fidler:netcalcguide}. The model enables us to express also more advanced fork-join systems and networks in the following sections.}\cond{short}{Compared to~\cite{rizk:forkjoin}, we define a general queueing model in max-plus system theory~\cite{baccelli:synchronizationlinearity, chang:performanceguarantees} that enables us to express also more advanced fork-join systems and networks in the following sections.} Further, we generalize central results from~\cite{rizk:forkjoin} to the case of G$\mid$G$\mid$1 queues. \cond{long}{For the special case of GI$\mid$GI$\mid$1 queues we recover the bounds from~\cite{rizk:forkjoin}.}
\subsection{Fork-Join System Model}
We label jobs by $n$ where $n \ge 1$ and denote $A(n)$ the time of arrival of job $n$. For notational convenience, we define $A(0) = 0$. Further, we let $A(m,n) = A(n) - A(m)$ be the time between the arrival of job $m$ and job $n$ for $n \ge m \ge 1$. It follows that $A(\nu,\nu+1)$ denotes the inter-arrival time between job $\nu$ and job $\nu+1$ for $\nu \ge 1$, so that we can also express
\begin{equation}
A(m,n) = \sum_{\nu=m}^{n-1} A(\nu,\nu+1)
\label{eq:arrivalincrements}
\end{equation}
as a sum of inter-arrival times. By definition, $A(n,n)=0$ for $n \ge 1$. Similarly, $D(n)$ defines the departure time of job $n$.

We denote $S(n)$ the service time of job $n$ for $n \ge 1$. The cumulative service time for jobs $m$ up to and including $n$ follows for $n \ge m \ge 1$ as
\begin{equation}
S(m,n) = \sum_{\nu=m}^n S(\nu) .
\label{eq:serviceincrements}
\end{equation}
For a lossless, work-conserving, first-in first-out (fifo) system, the departures for all $n \ge 1$ are~\cite{chang:performanceguarantees}
\begin{equation}
D(n) = \max_{\nu \in [1,n]} \{ A(\nu) + S(\nu,n) \} .
\label{eq:exactmaxplusserviceprocess}
\end{equation}
To see this, consider all jobs $[\nu, n]$ that belong to the same busy period as job $n$, i.e., job $\nu$ arrives at time $A(\nu)$ at an empty system that is continuously busy afterwards. It follows that $D(n) = A(\nu) + S(\nu,n)$ since starting at $A(\nu)$ all jobs $[\nu,n]$ have to be served until job $n$ departs, requiring a cumulative service time of $S(\nu,n)$.
In general $\nu$ is a priori unknown, however, it is known that there exists $\nu \ge 1$ such that $D(n) = A(\nu) + S(\nu,n)$ for $n \ge 1$ so that it can be concluded that
%
%
\begin{equation}
D(n) \le \max_{\nu \in [1,n]} \{ A(\nu) + S(\nu,n) \} .
\label{eq:maxplusserviceprocess}
\end{equation}
Further, $D(n) \ge D(\nu) + S(\nu+1,n)$ for all $n > \nu \ge 1$. Since $D(\nu) \ge A(\nu) + S(\nu)$ it follows that $D(n) \ge A(\nu) + S(\nu,n)$ for all $n \ge \nu \ge 1$. Consequently, $D(n) \ge \max_{\nu \in [1,n]} \{ A(\nu) + S(\nu,n)\}$ so that combined with~\eqref{eq:maxplusserviceprocess} we obtain~\eqref{eq:exactmaxplusserviceprocess}. For many applications the upper bound~\eqref{eq:maxplusserviceprocess} is sufficient as it will provide upper bounds of waiting and sojourn times.

For the sojourn time of job $n$ defined as $T(n) = D(n) - A(n)$ it follows by insertion of~\eqref{eq:exactmaxplusserviceprocess} for $n \ge 1$ that
\begin{equation}
T(n) = \max_{\nu \in [1,n]} \{ S(\nu,n) - A(\nu,n) \} .
\label{eq:sojourntime}
\end{equation}
For the waiting time, defined as $W(n) = [D(n-1) - A(n)]^+$ for $n \ge 1$, where $[X]^+ = \max\{X,0\}$ is the non-negative part and $D(0)=0$ by definition, it holds for $n \ge 1$ that
\begin{equation}
W(n) = \biggl[ \sup_{\nu \in [1,n-1]} \{ S(\nu,n-1) - A(\nu,n) \} \biggr]^+ .
\label{eq:waitingtime}
\end{equation}
Here, we use $\sup$ since for $n=1$~\eqref{eq:waitingtime} evaluates an empty set. For non-negative real numbers the $\sup$ of an empty set is zero.

\cond{long}{
\begin{figure}
  \centering
  \includegraphics[width=0.85\columnwidth]{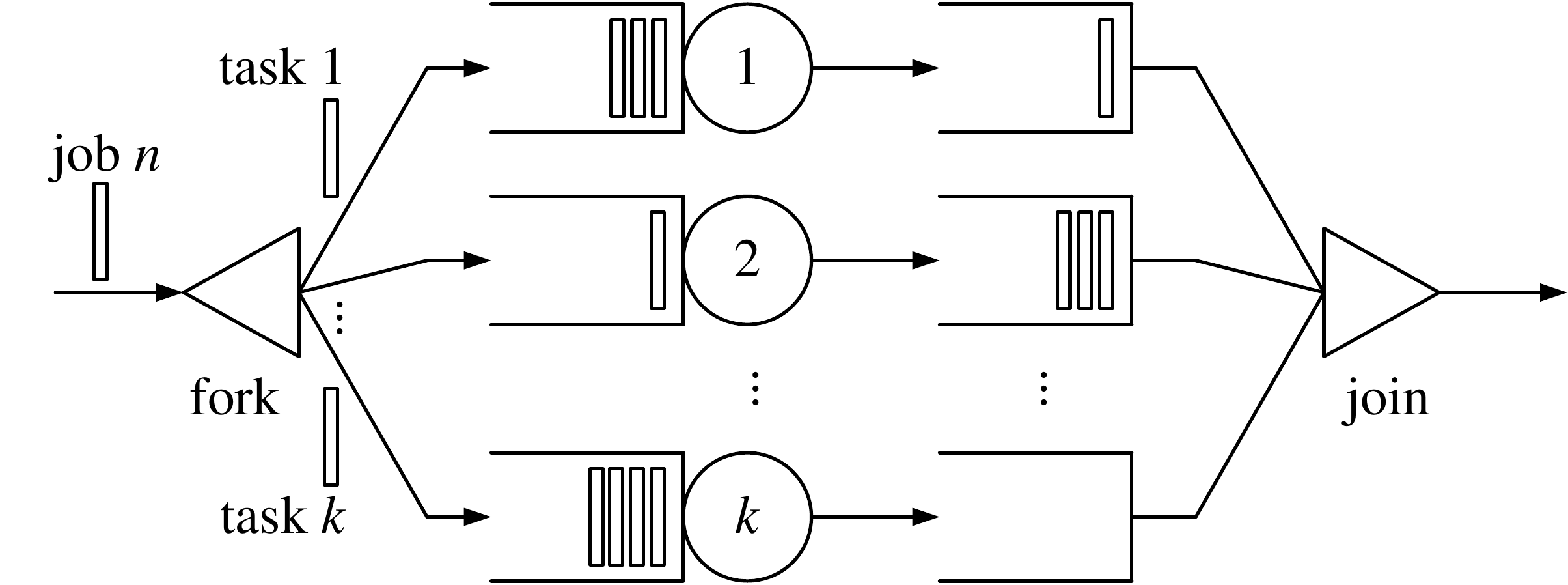}
  \caption{Fork-join system. Each job is composed of $k$ tasks with individual service requirements, that are mapped to $k$ fifo servers (fork). Once all tasks of a job are completed, the job leaves the system (join), i.e., the tasks of a job wait at the join stage until all tasks of the job are completed.}
  \label{fig:fjqueue}
\end{figure}
}
In a fork-join system, \cond{long}{see Fig.~\ref{fig:fjqueue}, }each job $n \ge 1$ is composed of $k$ tasks with service times $S_i(n)$ for $i \in [1,k]$, i.e., the service requirements of the tasks may differ from each other. The tasks are mapped to $k$ parallel servers (fork) and once all tasks are served, the job leaves the system (join). The parallel servers are not synchronized, i.e., server $i \in [1,k]$ starts serving task $i$ of job $n+1$ if any, once it finished serving task $i$ of job $n$ that departs from server $i$ at $D_i(n)$. A job $n$ is finished once all of its tasks $i \in [1,k]$ are finished, i.e.,
\begin{equation*}
D(n) = \max_{i \in [1,k]} \{ D_i(n) \} .
\end{equation*}
Hence, the sojourn time is $T(n) = \max_{i \in [1,k]} \{ D_i(n) - A(n) \}$ for $n \ge 1$ and by insertion of~\eqref{eq:exactmaxplusserviceprocess} for each server $i \in [1,k]$
\begin{equation}
T(n) = \max_{i \in [1,k]} \left\{ \max_{\nu \in [1,n]} \{ S_i(\nu,n) - A(\nu,n) \} \right\} .
\label{eq:nonblockingsojourntime}
\end{equation}
Similarly, a waiting time that considers the task that starts service last can be defined as a $\max_{i \in [1,k]}$ of~\eqref{eq:waitingtime}.

\cond{long}{
\begin{figure}
  \centering
  \includegraphics[width=0.77\columnwidth]{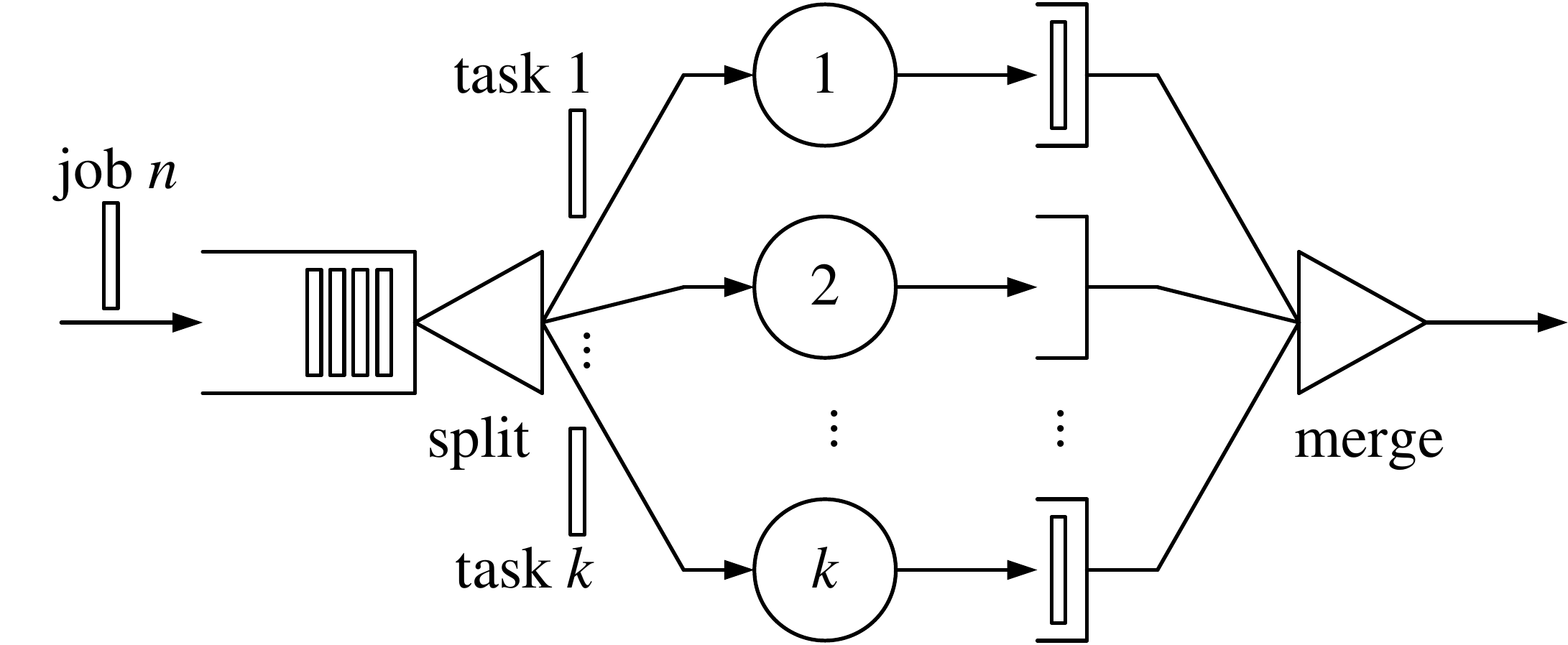}
  \caption{Split-merge system. Compared to the fork-join system, tasks have an additional synchronization constraint, i.e., the execution of the tasks of a job has to start at the same time.}
  \label{fig:smqueue}
\end{figure}
}
In a split-merge system on the other hand, \cond{long}{see Fig.~\ref{fig:smqueue}, }the execution of all tasks of a job starts simultaneously, i.e., once server $i$ finishes task $i$ of job $n$ it idles until all tasks $j \in [1,k]$ of that job are finished, before execution of the tasks of job $n+1$ if any starts. The combined service and idle times at each of the servers follows as $\max_{i \in [1,k]} \{ S_i(n) \}$ such that the split-merge system as a whole can be expressed as a single server system that is governed by~\eqref{eq:exactmaxplusserviceprocess} with service process
\begin{equation}
S(m,n) = \sum_{\nu=m}^n \max_{i \in [1,k]} \{ S_i(\nu) \} .
\label{eq:blockingserviceprocess}
\end{equation}

The general problem of split-merge systems is that $\max_{i \in [1,k]} \{ S_i(\nu) \}$ increases with $k$, with few exceptions such as in case of identical task service times. The increase implies longer idle times that result in a reduced stability region.
%
%
\subsection{Performance Bounds for GI$\mid$GI$\mid$1 and G$\mid$G$\mid$1 servers}
Next, we derive statistical performance bounds for the fork-join systems defined above. We consider the general model of G$\mid$G$\mid$1 servers. The results enable us to generalize recent findings obtained for iid service times, i.e., for a GI service model, in~\cite{rizk:forkjoin}. Like~\cite{rizk:forkjoin}, we assume that the inter-arrival times of jobs are independent of their service times. We do, however, not require that the service times of the tasks of a job are independent. We will show how to benefit from statistically independent tasks in Sec.~\ref{sec:klforkjoin}.

We consider arrival and service processes that belong to the broad class of $(\sigma,\rho)$-constrained processes~\cite{chang:performanceguarantees}, that are characterized by affine bounding functions of the moment generating function (MGF). The MGF of a random variable $X$ is defined as $\mathsf{M}_X(\theta) =\mathsf{E}[e^{\theta X}]$ where $\theta$ is a free parameter. The following definition adapts~\cite{chang:performanceguarantees} to max-plus systems.
\begin{definition}
\label{def:sigmarho}
An arrival process is $(\sigma_A,\rho_A)$-lower constrained if for all $n \ge m \ge 1$ and $\theta > 0$ it holds that
\begin{equation*}
\mathsf{E}[e^{-\theta A(m,n)}] \le e^{-\theta (\rho_A(-\theta) (n-m) - \sigma_A(-\theta))} .
\end{equation*}
Similarly, a service process is $(\sigma_S,\rho_S)$-upper constrained if for all $n \ge m \ge 1$ and $\theta > 0$ it holds that
\begin{equation*}
\mathsf{E}[e^{\theta S(m,n)}] \le e^{\theta(\rho_S(\theta) (n-m+1) + \sigma_S(\theta))} .
\end{equation*}
\end{definition}

For the special case of iid inter-arrival and service times, i.e., in case of GI$\mid$GI$\mid$1 servers, we have $\mathsf{E}[e^{-\theta A(n,n+1)}] = \mathsf{E}[e^{-\theta A(1,2)}]$ and $\mathsf{E}[e^{\theta S(n)}] = \mathsf{E}[e^{\theta S(1)}]$ for all $n \ge 1$. Further, since the MGF of a sum of independent random variables is the product of their MGFs, i.e., $\mathsf{M}_{X+Y}(\theta) = \mathsf{M}_{X}(\theta)\mathsf{M}_{Y}(\theta)$, minimal traffic and service parameters can be derived from~\eqref{eq:arrivalincrements} and~\eqref{eq:serviceincrements} as $\sigma_A(-\theta) = 0$, $\sigma_S(\theta) = 0$,
\begin{equation}
\rho_A(-\theta) = -\frac{1}{\theta} \ln \mathsf{E}\bigl[e^{-\theta A(1,2)}\bigr],
\label{eq:arrivalparameter}
\end{equation}
and
\begin{equation}
\rho_S(\theta) = \frac{1}{\theta} \ln \mathsf{E}\bigl[e^{\theta S(1)}\bigr].
\label{eq:serviceparameter}
\end{equation}
Parameter $\rho_A(-\theta)$ decreases with $\theta > 0$ from the mean to the minimum inter-arrival time and $\rho_S(\theta)$ increases from the mean to the maximum service time.

\begin{theorem}
\label{th:gg1}
Consider a fork-join system with $i \in [1,k]$ parallel servers, each with arrivals and service as specified by Def.~\ref{def:sigmarho}. It holds for the waiting time for all $n \ge 1$ that
\begin{equation*}
\mathsf{P} \left[ W(n) > \tau \right] \le \sum_{i=1}^k \alpha_i e^{-\theta_i \tau},
\end{equation*}
and for the sojourn time that
\begin{equation*}
\mathsf{P} [ T(n) > \tau ] \le \sum_{i=1}^k \alpha_i e^{\theta_i \rho_{S_i}(\theta_i)} e^{-\theta_i \tau} .
\end{equation*}
In the general case of G$\mid$G$\mid$1 servers, the free parameters \mbox{$\theta_i > 0$} have to satisfy $\rho_{S_i}(\theta_i) < \rho_A(-\theta_i)$ for $i \in [1,k]$ and
\begin{equation*}
\alpha_i = \frac{e^{\theta_i(\sigma_A(-\theta_i) + \sigma_{S_i}(\theta_i))}}{1-e^{-\theta_i (\rho_A(-\theta_i)-\rho_{S_i}(\theta_i))}} .
\end{equation*}
In the special case of GI$\mid$GI$\mid$1 servers, the $\theta_i > 0$ have to satisfy $\rho_{S_i}(\theta_i) \le \rho_A(-\theta_i)$ and $\alpha_i=1$ for $i\in[1,k]$.
\end{theorem}
For the special case of GI$\mid$GI$\mid$1 servers, Th.~\ref{th:gg1} recovers the bounds from~\cite{rizk:forkjoin}. Further, for $k=1$ the classical bound for the waiting time of a single GI$\mid$GI$\mid$1 server~\cite{kingman:gg1} is obtained in the max-plus system theory. Like~\cite{kingman:gg1}, the proof uses Doob's martingale inequality~\cite{doob:stochasticprocesses}. The proof for the G$\mid$G$\mid$1 server brings the approach from~\cite{chang:performanceguarantees, fidler:momentcalculus} forward to max-plus fork-join systems. The important property of the G$\mid$G$\mid$1 result is that it differs only by constants $\alpha_i$ from the GI$\mid$GI$\mid$1 result and otherwise recovers the characteristic exponential tail decay $e^{-\theta_i \tau}$ with decay rate $\theta_i$. We note that Th.~\ref{th:gg1} does not make an assumption of independence regarding the parallel servers. Indeed, independence cannot be assumed as the waiting and sojourn times of the individual servers depend on the common arrival process~\cite{baccelli:forkjoin, kemper:forkjoin}. The proof is provided in the appendix.

To obtain a solution for split-merge systems, recall that the system can be expressed as a single server with service process~\eqref{eq:blockingserviceprocess}. A corresponding service constraint as in Def.~\ref{def:sigmarho} can be directly inserted into Th.~\ref{th:gg1} where $k=1$. For iid service times it can be seen from~\eqref{eq:blockingserviceprocess} and~\eqref{eq:serviceparameter} that \cond{short}{$\rho_S(\theta)$}\cond{long}{
\begin{equation*}
\rho_S(\theta) = \frac{1}{\theta} \ln \mathsf{E}\bigl[e^{\theta \max_{i \in [1,k]} \{S_i(1)\}}\bigr] \le \frac{1}{\theta} \ln \biggl(\sum_{i = 1}^{k} \mathsf{E}\bigl[e^{\theta S_i(1)}\bigr] \biggr).
\end{equation*}}
has at most a logarithmic growth with the number of parallel servers, resulting in a corresponding reduction of the stability region. A decrease of the stability region with $\ln k$ is also shown in~\cite{rizk:forkjoin}, where the authors advise against split-merge implementations based on an in-depth comparison with fork-join systems. In the sequel, we will focus on fork-join systems.

To investigate the scaling of fork-join systems, we consider the homogeneous case where the service times $S_i(\nu)$ for $i \in [1,k]$ are identically distributed\cond{long}{ with parameters $\rho_{S_i} = \rho_S$}, but not necessarily independent. From Th.~\ref{th:gg1} it follows for the sojourn time that
\begin{equation}
\mathsf{P} [ T(n) > \tau ] \le k \alpha e^{\theta \rho_{S}(\theta)} e^{-\theta \tau} .
\label{eq:homogeneousparallelsojourntime}
\end{equation}
After some reordering, the bound of the sojourn time
\begin{equation*}
\mathsf{P} \left[ T(n) > \tau + \frac{\ln k}{\theta} \right] \le \alpha e^{\theta \rho_{S}(\theta)} e^{-\theta \tau}
\end{equation*}
shows a growth with $\ln k$. By integration of the tail of~\eqref{eq:homogeneousparallelsojourntime}, where we use that $\mathsf{P} [ T(n) > \tau ] \le 1$,
\cond{long}{
\begin{align*}
\mathsf{E}[T(n)] = & \int_{0}^{\infty} \mathsf{P} [ T(n) > \tau ] \mathrm{d}\tau \\
\le & \int_{0}^{\tau^*} \mathrm{d}\tau + k \alpha e^{\theta \rho_{S}(\theta)} \int_{\tau^*}^{\infty} e^{-\theta \tau} \mathrm{d}\tau,
\end{align*}
and $\tau^* = \ln \bigl(k \alpha e^{\theta \rho_{S}(\theta)}\bigr)/\theta$, }the expected sojourn time
\begin{equation}
\mathsf{E} [ T(n) ] \le \rho_{S}(\theta) +  \frac{\ln(k\alpha)+1}{\theta}
\label{eq:homogeneousparallelexpectedsojourntime}
\end{equation}
is also limited by $\ln k$. The result applies for general arrival and service processes and generalizes the finding of $\ln k$ that is obtained in~\cite{baccelli:forkjoin, rizk:forkjoin} for iid service times. The growth is larger for smaller $\theta$ corresponding to a higher utilization.

\paragraph*{M$\mid$M$\mid$1 Servers}
For illustration, we consider iid exponential inter-arrival and service times with parameters $\lambda$ and $\mu$, respectively. The normalized log-MGFs~\eqref{eq:arrivalparameter} and~\eqref{eq:serviceparameter} are
\begin{equation}
\rho_A(-\theta) = -\frac{1}{\theta} \ln \!\left(\!\frac{\lambda}{\lambda+\theta}\!\right)
\text{ and }
\rho_{S}(\theta) = \frac{1}{\theta} \ln \!\left(\!\frac{\mu}{\mu-\theta}\!\right)\!,
\label{eq:mm1parameters}
\end{equation}
\begin{figure}
  \centering
  \cond{short}{\includegraphics[width=0.6\columnwidth]{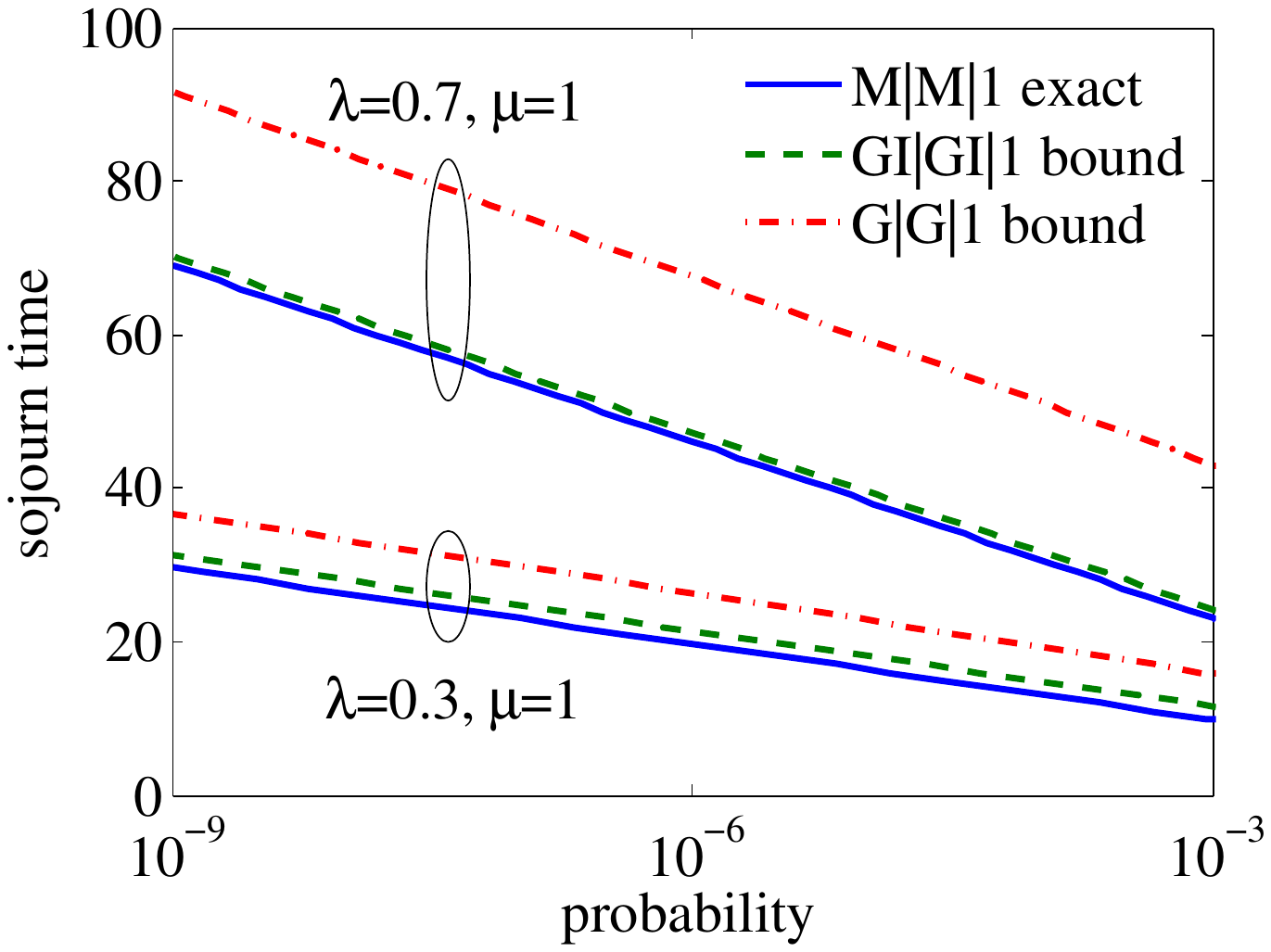}}
  \cond{long}{\includegraphics[width=0.75\columnwidth]{mm1boundsmu1lambda0307}}
  \caption{M$\mid$M$\mid$1 server. The bounds show the correct exponential tail decay.}
  \label{fig:backlog_qt}
  \cond{short}{\vspace{-10pt}}
\end{figure}
where $\theta \in (0,\mu)$. From the condition $\rho_{S}(\theta) \le \rho_A(-\theta)$ it follows that $\theta \le \mu-\lambda$ under the stability condition $\mu > \lambda$. By choice of the maximal $\theta = \mu - \lambda$ we have from \eqref{eq:homogeneousparallelsojourntime} that
\begin{equation}
\mathsf{P} \left[ T(n) > \tau \right] \le k \frac{\mu}{\lambda} e^{-(\mu-\lambda) \tau} .
\label{eq:mm1bound}
\end{equation}

First, we consider the accuracy of the bounds from Th.~\ref{th:gg1} for $k=1$ that is the case of a single M$\mid$M$\mid$1 server. Compared to the exact distribution of the sojourn time of the M$\mid$M$\mid$1 server, that is $\mathsf{P} \left[ T(n) > \tau \right] \le e^{-(\mu-\lambda) \tau}$ see, e.g.,~\cite{adan:queueingsystems},~\eqref{eq:mm1bound} has the same tail decay, but shows a different pre-factor, i.e., $\mu/\lambda$. Obviously, the bound becomes better if the utilization is high so that $\mu/\lambda$ approaches one. In Fig.~\ref{fig:backlog_qt}, we illustrate the bounds from Th.~\ref{th:gg1} for a single server compared to the exact M$\mid$M$\mid$1 result. Clearly, the curves show the same tail decay, where the GI$\mid$GI$\mid$1 bound provides better numerical accuracy compared to the G$\mid$G$\mid$1 bound that does not use independence of the increment processes and hence has parameter $\alpha > 1$.

In Fig.~\ref{fig:mm1forkjoin}, we consider $k \ge 1$ parallel servers with $\mu_i=1$ and show bounds of the expected sojourn time and sojourn time quantiles $\tau$, where $\mathsf{P}[T(n) \ge \tau] \le \varepsilon$ and $\varepsilon = 10^{-6}$. The curves show the characteristic logarithmic growth with $k$.

For reasons of space, we will frequently employ the M$\mid$M$\mid$1 model~\eqref{eq:mm1parameters} for numerical evaluations. We note that results of the same type can be derived from Th.~\ref{th:gg1} for G$\mid$G$\mid$1 servers with parameters specified by Def.~\ref{def:sigmarho}.
%
%
\subsection{Load Balancing}
\label{sec:loadbalacingforkjoin}
Next, we consider heterogeneous GI$\mid$GI$\mid$1 servers, i.e., the tasks of a job can have service requirements $S_i(n)$ for $i \in [1,k]$ with different parameters. Also, the parallel servers can have different capacities $c_i$. The resulting service times become $S_i(n)/c_i$ and the MGF follows for $n \ge 1$ as
\begin{equation}
\mathsf{M}_{S_i(n)/c_i} (\theta_i) = \mathsf{E}[e^{\theta_i S_i(1)/c_i}] = \mathsf{M}_{S_i(1)} (\theta_i/c_i) .
\label{eq:mgfscaling}
\end{equation}
By application of Th.~\ref{th:gg1} it holds for all $n \ge 1$ that
\begin{equation}
\mathsf{P} [ T(n) > \tau ] \le \sum_{i=1}^k e^{\theta_i \rho_{S_i/c_i}(\theta_i)} e^{-\theta_i \tau} ,
\label{eq:sojourntimeforkjoin}
\end{equation}
for all $\theta_i > 0$ that satisfy $\rho_{S_i/c_i}(\theta_i) \le \rho_{A}(-\theta_i)$.
\begin{figure}
  \centering
  \cond{short}{\includegraphics[width=0.6\columnwidth]{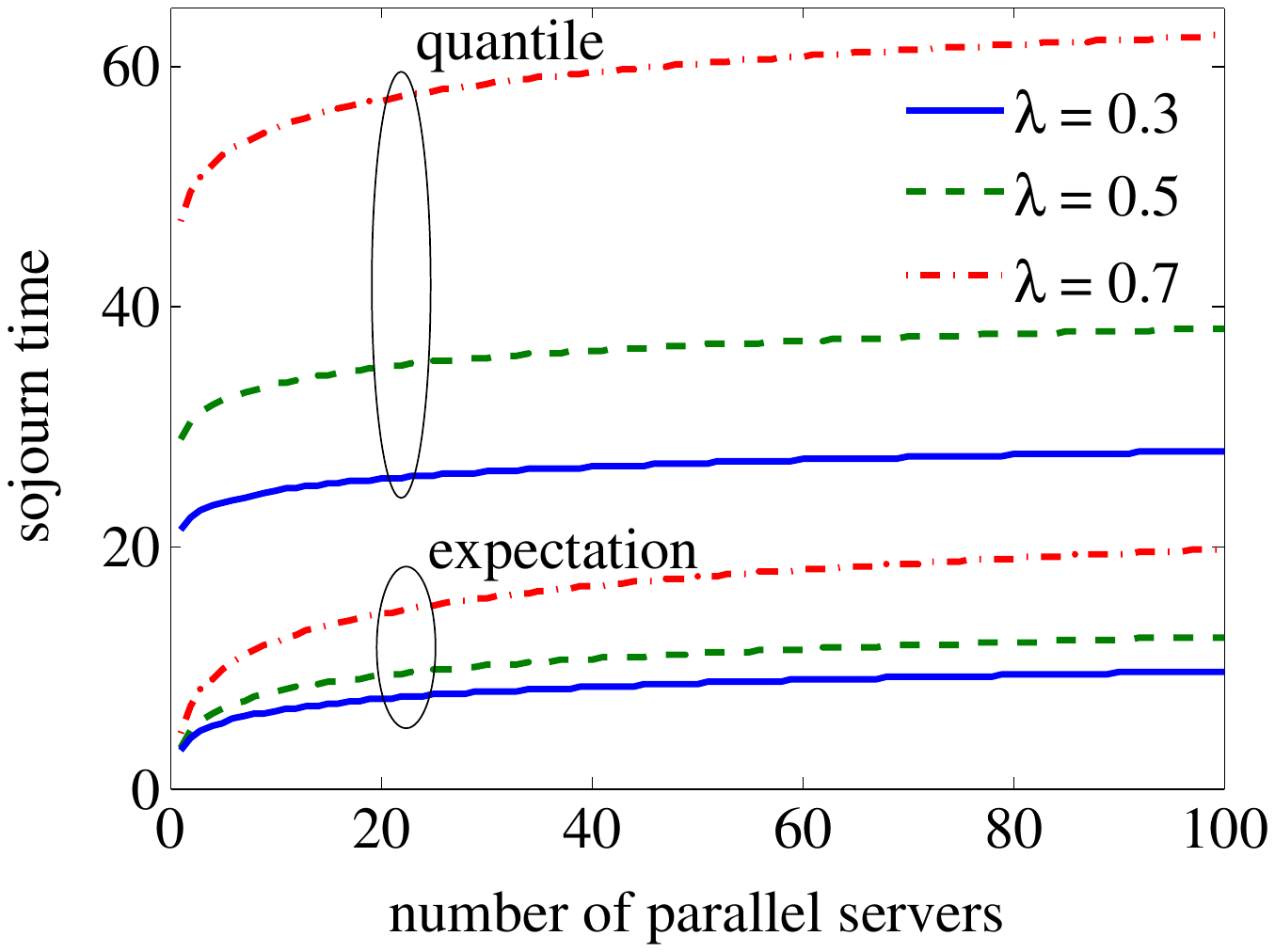}}
  \cond{long}{\includegraphics[width=0.75\columnwidth]{mm1forkjoinmu1lambda030507}}
  \caption{Fork-join system. Sojourn time bounds grow with $\ln k$ for $k$ servers.}
  \label{fig:mm1forkjoin}
  \cond{short}{\vspace{-10pt}}
\end{figure}

In the following we consider two different strategies to assign the capacities $c_i$ for load balancing. For ease of exposition, we assume that the $c_i$ are continuous.
\begin{enumerate}
  \item Assign capacity proportional to the mean service time requirements of the tasks. As a consequence, the mean service times and the utilization of each of the parallel servers becomes identical. The allocation is independent of the inter-arrival times. It is stable as long as the mean inter-arrival time is larger than the mean service time.
      \label{strt:mean}
  \item Assign capacity such that the same statistical bound of the sojourn time is achieved for all servers. This allocation may differ from the first one, so that the parallel servers may not have identical utilization.
      \label{strt:tail}
\end{enumerate}
\paragraph*{M$\mid$M$\mid$1 Servers}
We start with the basic case of exponential inter-arrival and service times with parameters $\lambda$ and $\mu_i$, respectively. In case of strategy~\ref{strt:mean}, capacity is allocated proportionally to the mean service time requirements, i.e., $c_i = \upsilon/\mu_i$ where $\upsilon$ is a positive constant. Given the sum of the capacities is limited as $c = \sum_{i=1}^k c_i$, we have $\upsilon = c/\sum_{i=1}^{k} 1/\mu_i$. By insertion of~\eqref{eq:mgfscaling} into~\eqref{eq:serviceparameter} it follows with~\eqref{eq:mm1parameters} that
\begin{equation*}
\rho_{S_i/c_i}(\theta_i) = \frac{1}{\theta_i} \ln \left(\frac{\mu_i}{\mu_i - \theta_i/c_i}\right) = \frac{1}{\theta_i} \ln \left(\frac{\upsilon}{\upsilon-\theta_i}\right) .
\end{equation*}
To ensure $\rho_{S_i/c_i}(\theta_i) \le \rho_{A}(-\theta_i)$, where $\rho_{A}(-\theta_i)$ is given by~\eqref{eq:mm1parameters}, it has to hold that $\theta_i \le \upsilon-\lambda$ under the stability condition that $\upsilon > \lambda$. By choice of the maximal $\theta_i = \upsilon-\lambda$,~\eqref{eq:sojourntimeforkjoin} evaluates to
\begin{equation*}
\mathsf{P} [ T(n) > \tau ] \le k \frac{\upsilon}{\lambda} e^{-(\upsilon-\lambda) \tau}.
\end{equation*}
The result shows that strategy~\ref{strt:mean} in case of M$\mid$M$\mid$1 servers also achieves the target of strategy~\ref{strt:tail}, i.e., it implements identical bounds of the sojourn time for all $k$ servers.
\paragraph*{Gaussian Servers}
Secondly, we consider Gaussian inter-arrival and service times with mean and variance $\eta_A,\varsigma^2_A$ and $\eta_{S},\varsigma^2_{S}$, respectively. The normalized log-MGFs~\eqref{eq:arrivalparameter} and~\eqref{eq:serviceparameter} are \cond{short}{\vspace{-6pt}}
\begin{equation*}
\rho_A(-\theta) = \eta_A - \frac{\theta}{2}\varsigma^2_A  \quad \text{and} \quad \rho_{S}(\theta) = \eta_{S} + \frac{\theta}{2}\varsigma^2_{S} .
\end{equation*}
Considering heterogeneous service parameters $\eta_{S_i},\varsigma^2_{S_i}$, strategy~\ref{strt:mean} implements a capacity allocation $c_i = \upsilon \eta_{S_i}$, where $\upsilon$ is a positive constant as above. With~\eqref{eq:mgfscaling} we obtain
\begin{equation*}
\rho_{S_i/c_i}(\theta_i) = \frac{\eta_{S_i}}{c_i} + \frac{\theta_i}{2}\left(\frac{\varsigma_{S_i}}{c_i}\right)^2 = \frac{1}{\upsilon}+ \frac{\theta_i}{2} \left(\frac{\varsigma_{S_i}}{\upsilon \eta_{S_i}}\right)^2 .
\end{equation*}
A bound of the sojourn time follows from~\eqref{eq:sojourntimeforkjoin} where we choose the maximal $\theta_i$ that satisfy $\rho_{S_i/c_i}(\theta_i) \le  \rho_A(-\theta_i)$ for all $i \in [1,k]$. Clearly, as the parameters $\rho_{S_i/c_i}(\theta_i)$ are heterogeneous, the maximal tail decay parameters $\theta_i$ will in general not be identical. As a consequence, the capacity allocation does in this case not achieve the goal of strategy~\ref{strt:tail}.

To implement strategy~\ref{strt:tail},~\eqref{eq:sojourntimeforkjoin} implies that the capacity is allocated in a way such that identical parameters $\rho_{S_i/c_i}(\theta_i)$ are achieved for the same target tail decay parameter $\theta_i = \theta$ for all servers $i \in [1,k]$. Considering the minimal capacity allocation that satisfies $\rho_{S_i/c_i}(\theta) \le \rho_A(-\theta)$, the capacities $c_i$ follow readily as the solution of the quadratic problem
\begin{equation*}
\frac{\eta_{S_i}}{c_i} + \frac{\theta}{2} \left(\frac{\varsigma_{S_i}}{c_i}\right)^2 = \eta_A - \frac{\theta}{2}\varsigma^2_A .
\end{equation*}
The important observation is that the obvious strategy~\ref{strt:mean}, that balances the average utilization, does in general not achieve the same bound of the sojourn time for all servers, i.e., under strategy~\ref{strt:mean} certain servers may be late frequently.
%
%
\section{Thinning and Resequencing}
\label{sec:thinning}
We investigate systems of $k$ parallel servers with thinning and resequencing, as it applies, e.g., in case of multi-path routing. The difference to fork-join systems is that jobs are not composed of tasks that are served in parallel but instead each job is mapped in its entirety to one of the parallel servers. As a consequence, the external arrival process $A(n)$ is divided into $k$ thinned processes $A_i(m)$. The numbering of jobs is such that $A(n)$ is the arrival time of job $n$ before thinning, whereas $A_i(m)$ denotes the arrival time of the $m$th job of the thinned process at system $i$ with service time $S_i(m)$. The departures $D_i(m)$ are resequenced in the original order of $A(n)$ to form the departure process $D(n)$. We contribute solutions for random as well as deterministic thinning.
\cond{long}{An example of a system with deterministic thinning is depicted in Fig.~\ref{fig:trqueue}
\begin{figure}
  \centering
  \includegraphics[width=0.95\columnwidth]{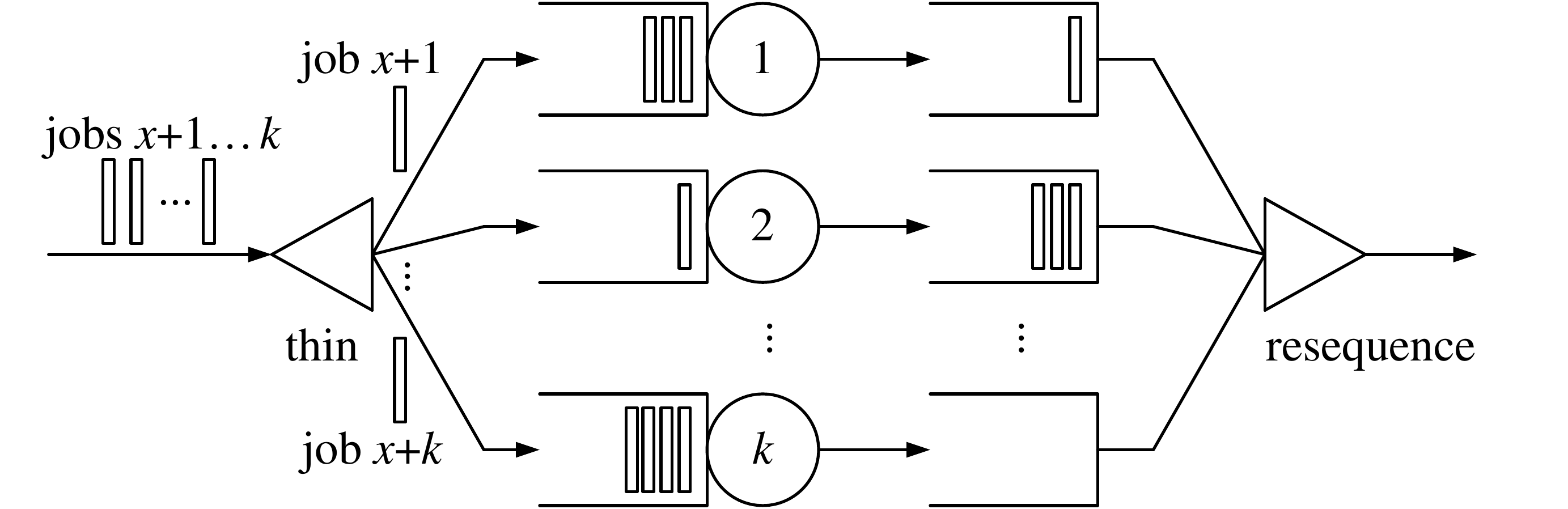}
  \caption{System with thinning and resequencing. Compared to a fork-join system, jobs are not composed of tasks. Instead, entire tasks are mapped to the parallel servers, e.g., deterministically in round robin order.}
  \label{fig:trqueue}
\end{figure}
}
%
%
\subsection{Random Thinning}
In case of random thinning, each job is mapped to one of the $k$ servers according to iid discrete (not necessarily uniform) random variables with support $[1,k]$. From the iid property it follows that the mapping of each job to a certain server $i \in [1,k]$ is an independent Bernoulli trial with parameter $p_i$. Denote $X_i(m)$ the number of the job that becomes the $m$th job that is mapped to server $i$. It follows that $X_i(m)$ is a sum of $m$ iid geometric random variables with parameter $p_i$, i.e., $X_i(m)$ is negative binomial. The arrival process of server $i$ is
\begin{equation}
A_i(m) = A(X_i(m)) ,
\label{eq:randomthinning}
\end{equation}
for $m \ge 1$. Conversely, given jobs $[1,n]$ of the external arrival process, denote $Y_i(n)$ the quantity of jobs that are mapped to server $i$. It follows that $Y_i(n)$ is binomial with parameter $p_i$. Considering fifo servers with arrival processes~\eqref{eq:randomthinning} and departure processes $D_i(m)$ for $i \in [1,k]$, the combined in-sequence departure process for $n \ge 0$ follows as
\begin{equation}
D(n) = \max_{i \in [1,k]} \{D_i(Y_i(n))\} ,
\label{eq:randomresequencing}
\end{equation}
where $D_i(0)=0$ by convention. Note that~\eqref{eq:randomresequencing} has to verify only the departure of job $Y_i(n)$ of each server $i \in [1,k]$ since the departure of job $Y_i(n)$ from server $i$ implies the departure of all jobs $\nu \in [1,Y_i(n)]$ of the same server due to fifo order.

Next, we substitute~\eqref{eq:exactmaxplusserviceprocess} for $D_i(Y_i(n))$ in~\eqref{eq:randomresequencing} and use~\eqref{eq:randomthinning} to derive the sojourn time $T(n) = D(n) - A(n)$ for $n \ge 1$ as
\begin{equation*}
T(n) = \max_{i \in [1,k]} \left\{ \sup_{\nu \in [1,Y_i(n)]} \{S_i(\nu,Y_i(n)) - A(X_i(\nu),n) \} \right\} ,
\end{equation*}
where $\sup\{\emptyset\} = 0$. Further, since $A(n) \ge A_i(Y_i(n))$ for all $i \in [1,k]$ and $n \ge 1$ and with~\eqref{eq:randomthinning}, we can estimate
\begin{equation}
T(n) \le \max_{i \in [1,k]} \left\{ \sup_{\nu \in [1,Y_i(n)]} \{S_i(\nu,Y_i(n)) - A_i(\nu,Y_i(n)) \} \right\} .
\label{eq:sojourntimethinning}
\end{equation}
Due to the similarity of~\eqref{eq:sojourntimethinning} with~\eqref{eq:nonblockingsojourntime}, the derivation of delay bounds closely follows the proof of Th.~\ref{th:gg1} and recovers the same result with one essential difference: Instead of the arrival process $A(n)$ the thinned processes $A_i(m)$ have to be considered for each of the servers $i \in [1,k]$, i.e., parameters $(\sigma_A,\rho_A)$ in Th.~\ref{th:gg1} are replaced by the parameters $(\sigma_{A_i},\rho_{A_i})$ of the thinned processes. The main effect of thinning is an increase of the parameters $\theta_i$ that result in a faster tail decay.

The thinned arrivals~\eqref{eq:randomthinning} are expressed as doubly random processes. Considering iid inter-arrival times the MGF of the thinned process is
\begin{equation*}
\mathsf{M}_{A_i(\nu,\nu+1)}(-\theta) = \mathsf{E} \left[(\mathsf{M}_{A(1,2)}(-\theta))^{X_i(1)} \right],
\end{equation*}
for $\nu \ge 1$. It follows after some reordering that
\begin{equation}
\mathsf{M}_{A_i(\nu,\nu+1)}(-\theta) = \mathsf{M}_{X_i(1)}(\ln \mathsf{M}_{A(1,2)}(-\theta)) .
\label{eq:doublystochasticarrivalmgf}
\end{equation}
Since $X_i(1)$ is a geometric random variable with MGF $\mathsf{M}_{X_i(1)}(\theta) = p_i e^{\theta}/ (1-(1-p_i)e^{\theta})$ for $\theta < -\ln(1-p_i)$, we obtain by insertion of~\eqref{eq:doublystochasticarrivalmgf} into~\eqref{eq:arrivalparameter} that
\begin{equation}
\rho_{A_i}(-\theta) = -\frac{1}{\theta} \ln \left( \frac{p_i \mathsf{M}_{A(1,2)}(-\theta)}{1-(1-p_i) \mathsf{M}_{A(1,2)}(-\theta) }\right) ,
\label{eq:splitinterarrivaltimes}
\end{equation}
where $\theta > 0$ so that $\mathsf{M}_{A(1,2)}(-\theta) < 1/(1-p_i)$.
%
%
\subsection{Deterministic Thinning}
In case of deterministic thinning, the random processes $X_i(m)$ and $Y_i(n)$ are replaced by deterministic functions. A round robin assignment of the jobs of an arrival process $A(n)$ to $k$ servers results in the processes $A_i(m)$ given by~\eqref{eq:randomthinning} where
\begin{equation}
X_i(m) = k(m-1)+i ,
\label{eq:splitarrivals}
\end{equation}
for $m \ge 1$ and $i \in [1,k]$. The combined in-sequence departure process $D(n)$ follows from~\eqref{eq:randomresequencing} with
\begin{equation}
Y_i(n) = \left\lceil \frac{n-i+1}{k} \right\rceil ,
\label{eq:mergeddepartures}
\end{equation}
for $n \ge 1$ and $i \in [1,k]$. To see this, note that job $n$ of the external arrival process becomes the $m=\lceil n/k\rceil$th job of server $j = (n-1) \! \mod k + 1$. Hence, $Y_i(n) = m$ for $i \le j$ and $Y_i(n) = m-1$ for $i > j$ due to the round robin procedure, as can be verified for~\eqref{eq:mergeddepartures}.

Statistical sojourn time bounds follow as in case of random thinning from~\eqref{eq:sojourntimethinning} as a variant of Th.~\ref{th:gg1}, where the parameters $(\sigma_{A_i}, \rho_{A_i})$ of the thinned arrival processes are used. Considering iid inter-arrival times, the MGFs of the deterministically thinned processes can be straightforwardly computed as $\mathsf{M}_{A_i(\nu,\nu+1)}(-\theta) = (\mathsf{M}_{A(1,2)}(-\theta))^k$ for $\nu \ge 1$. The parameters $\rho_{A_i}(-\theta)$ for $\theta > 0$ follow from~\eqref{eq:arrivalparameter} as
\begin{equation}
\rho_{A_i}(-\theta) = k \rho_A (-\theta) .
\label{eq:deterministicsplitinterarrivaltimes}
\end{equation}
For Markov arrival processes, an expression for deterministic thinning is given in~\cite{rizk:forkjoin}.
%
%
\paragraph*{$\text{E}_\text{k}$$\mid$M$\mid$1 Servers}
We consider arrivals $A(n)$ with iid exponential inter-arrival times with parameter $\lambda$. Deterministic thinning results in processes $A_i(n)$ where the inter-arrival times are a sum of $k$ exponential random variables that is Erlang-k distributed. It follows from~\eqref{eq:deterministicsplitinterarrivaltimes} for $\theta > 0$ that
\begin{equation*}
\rho_{A_i}(-\theta) = -\frac{k}{\theta} \ln \left(\frac{\lambda}{\lambda + \theta}\right).
\end{equation*}
\begin{figure}
  \centering
  \cond{short}{\includegraphics[width=0.6\columnwidth]{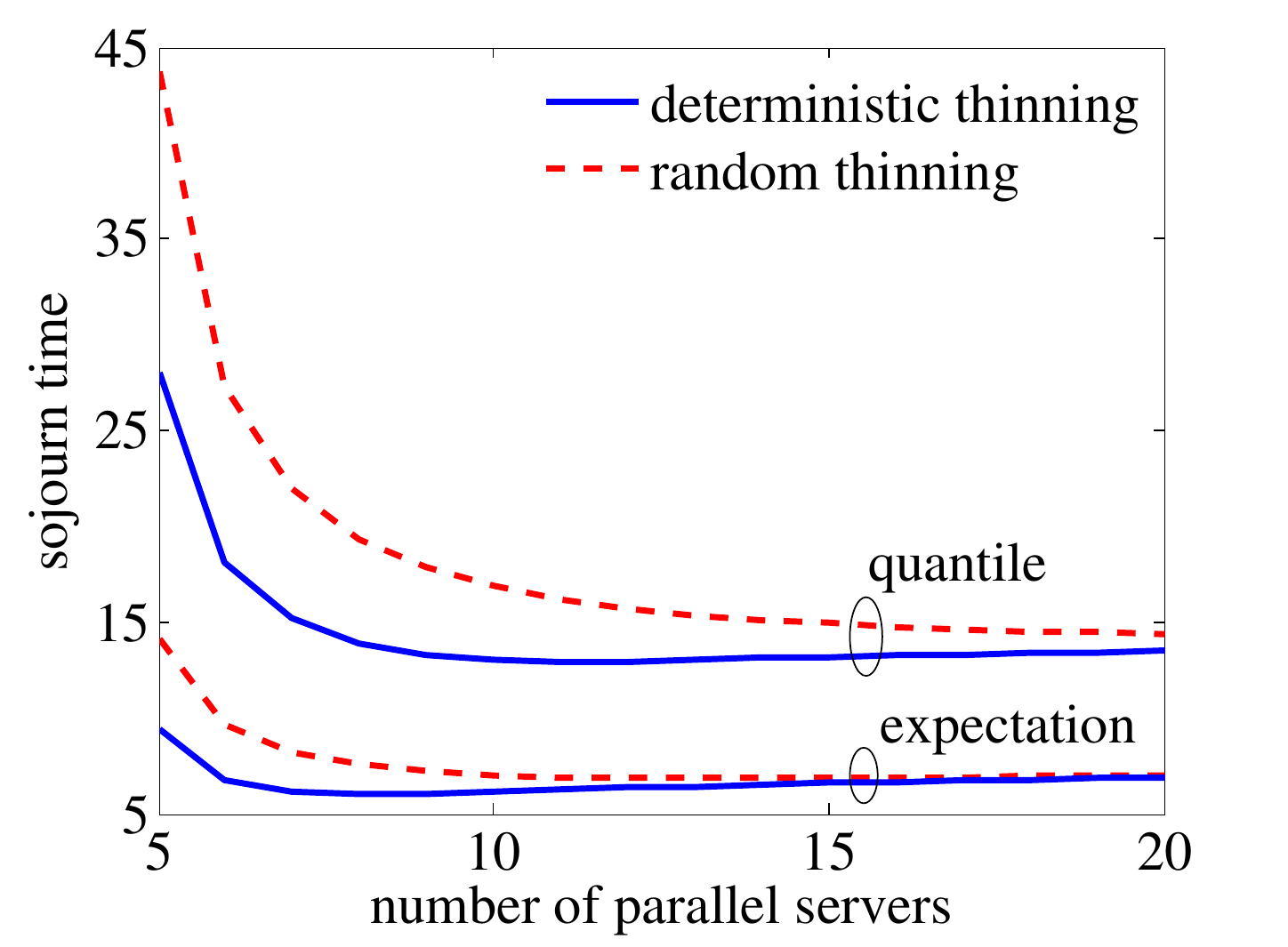}}
  \cond{long}{\includegraphics[width=0.75\columnwidth]{mm1thinningmu1lambda4}}
  \caption{Parallel servers achieve significant performance improvements if the utilization is high. Deterministic thinning outperforms random thinning.}
  \label{fig:thinning_mm1}
  \cond{short}{\vspace{-10pt}}
\end{figure}
In case of random thinning the inter-arrival times of the thinned processes are exponentially distributed with parameter $p_i \lambda$. With $p_i = 1/k$, we have from~\eqref{eq:splitinterarrivaltimes}  for $\theta > 0$ that
\begin{equation*}
\rho_{A_i}(-\theta) = - \frac{1}{\theta} \ln \left(\frac{\lambda}{\lambda + k \theta}\right) .
\end{equation*}
Given the service times at servers $i \in [1,k]$ are exponential with parameter $\mu_i = \mu$, we have $\rho_{S_i}(\theta) = \ln (\mu/(\mu-\theta))/\theta$ for $\theta \in (0,\mu)$. We choose the maximal parameter $\theta \in (0,\mu)$ such that $\rho_{S_i}(\theta) \le \rho_{A_i}(-\theta)$ for all $i \in [1,k]$ and obtain bounds of the sojourn time from~\eqref{eq:homogeneousparallelsojourntime} and~\eqref{eq:homogeneousparallelexpectedsojourntime}.

In Fig.~\ref{fig:thinning_mm1}, we compare deterministic and random thinning, where we show bounds of the expected sojourn time and of the sojourn time quantile with $\varepsilon = 10^{-3}$. The parameters are $\lambda=4$ and $\mu=1$, i.e., $k\ge 5$ parallel servers are required. Adding few more servers provides a significant advantage that is due to the decreasing utilization. The advantage diminishes, however, if $k$ becomes large and eventually the delay bounds start to grow due to potential waiting in the resequencing stage. \cond{long}{Note that Th.~\ref{th:gg1} does not assume independence of the parallel servers. }Deterministic thinning generally performs better than random thinning, where the advantage is larger in case of smaller $k$.
%
%
\subsection{Load Balancing}
We consider load balancing for heterogeneous servers and investigate two basic strategies: Divide the arrivals so that
\begin{enumerate}
  \item all servers have the same average utilization,
  \label{strt:splitmean}
  \item all servers have the same maximal tail decay.
  \label{strt:splittail}
\end{enumerate}
The strategies are similar to those that we considered for fork-join systems in Sec.~\ref{sec:loadbalacingforkjoin}. The difference is that we now thin the arrival process accordingly to achieve the goal.
\paragraph*{M$\mid$M$\mid$1 Servers}
We consider the case of exponential inter-arrival times with parameter $\lambda$ and exponential service times with parameters $\mu_i$ for servers $i \in [1,k]$. Random thinning is used to divide the arrivals into $k$ sub-processes with parameters $\lambda_i$ where $\sum_{i=1}^k \lambda_i=\lambda$. For strategy~\ref{strt:splitmean} the target is to achieve $\lambda_i/\mu_i = \lambda_j/\mu_j$ for all $i,j \in [1,k]$. It follows that $\lambda_i = \lambda \mu_i/(\sum_{j=1}^k \mu_j)$ for all $i \in [1,k]$.
Strategy~\ref{strt:splittail} seeks to equalize the maximal tail decay parameters $\theta_i = \mu_i - \lambda_i$, see~\eqref{eq:mm1parameters}, so that $\theta_i = \theta_j$ for all $i,j \in [1,k]$. A solution can be obtained iteratively as $\lambda_i = \mu_i - (\sum_{j=1}^k \mu_j - \lambda)/k$ for all $i \in [1,k]$. Servers that are assigned $\lambda_i < 0$ are excluded in the next iteration until all $\lambda_i$ are non-negative.

\begin{figure}
  \centering
  \hspace{-10pt}
  \subfigure[ $\lambda=0.4, \mu_1=1$]{
    \includegraphics[width=0.5\columnwidth]{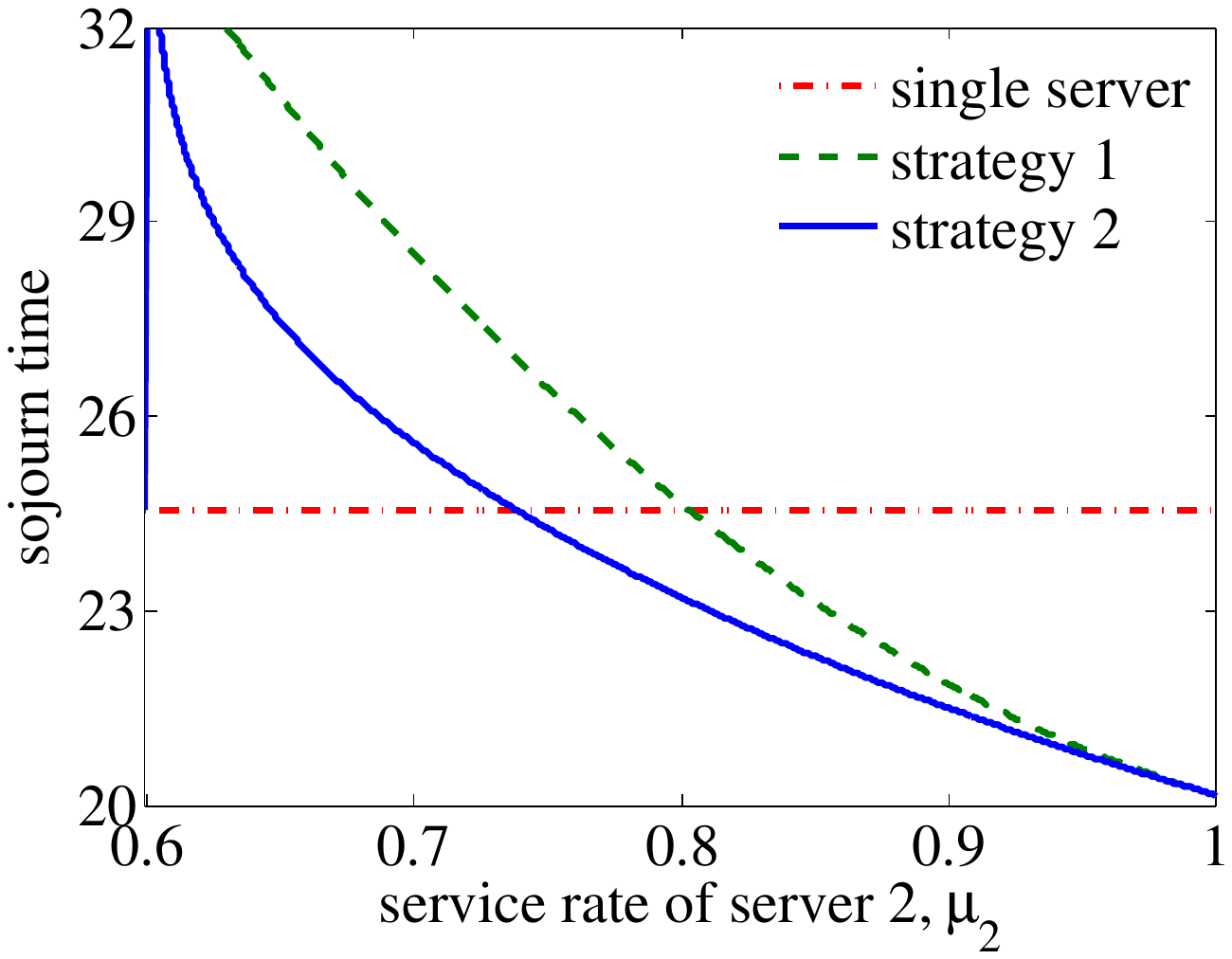}
  }
  \hspace{-15pt}
  \subfigure[ $\lambda=0.8, \mu_1=1$ ]{
    \includegraphics[width=0.5\columnwidth]{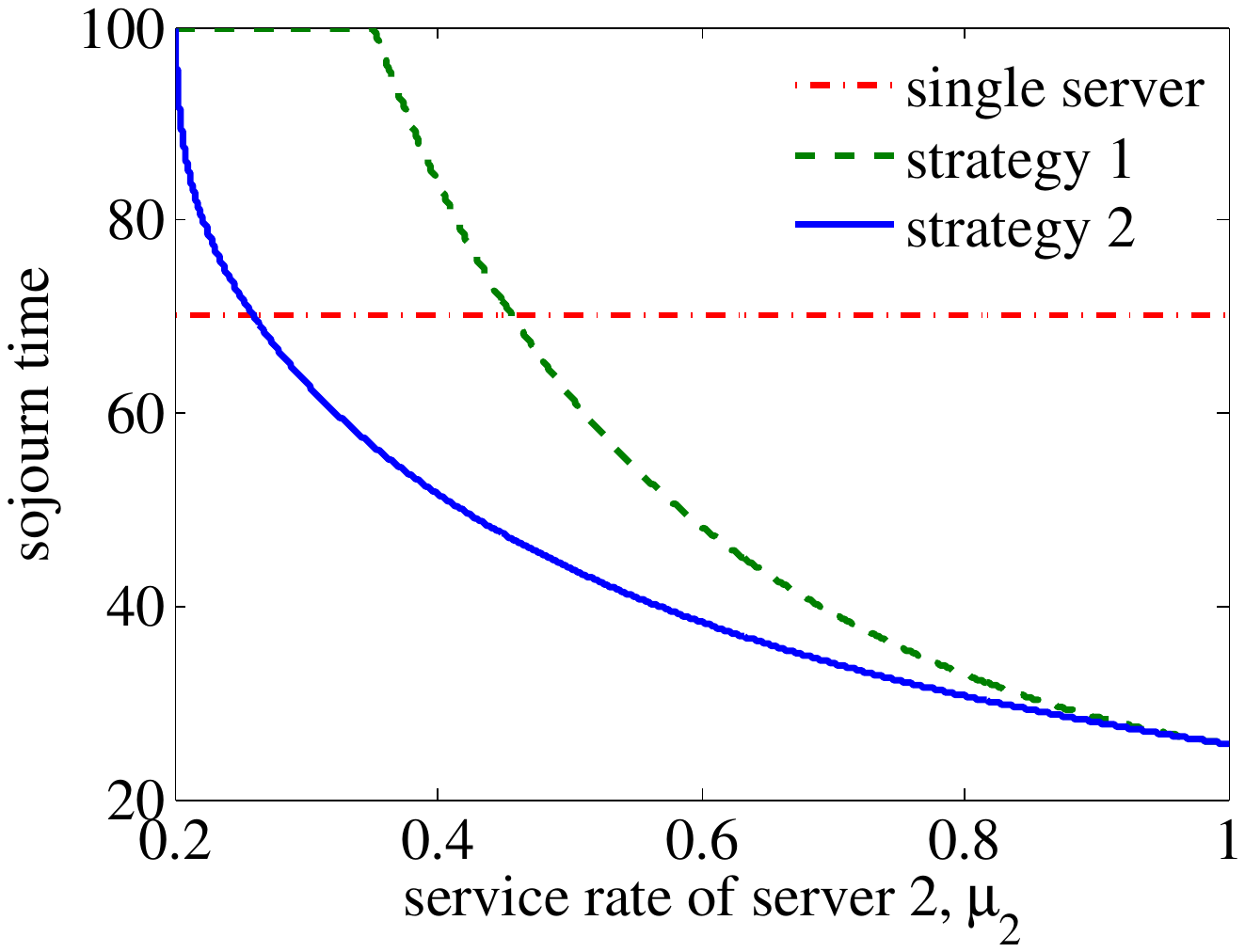}
  }
  \hspace{-10pt}
  \caption{The effectiveness of load balancing depends largely on the utilization.}
  \label{fig:thinning_mm1_loadbalancing}
  \cond{short}{\vspace{-10pt}}
\end{figure}
Fig.~\ref{fig:thinning_mm1_loadbalancing} shows bounds of the sojourn time for $\varepsilon = 10^{-6}$. We consider $k=2$ servers with parameters $\mu_1=1$, $\mu_2 \in [0.5,1]$, and $\lambda = 0.4$ and $0.8$, respectively. A practical application of the model is, e.g., the multi-path transmission of a packet data stream over Wifi and cellular in parallel. For comparison the sojourn time for the case without load balancing, where only server one is used, is shown. Clearly, load balancing performs very well if the utilization is high. On the other hand, an advantage is achieved only if the capacity of server two is sufficiently large, i.e., in certain cases it is better to exclude slow servers. Strategy~\ref{strt:splittail} performs better than strategy~\ref{strt:splitmean} as the allocation takes the tail decay of the sojourn time into account. If the servers are homogeneous, strategies~\ref{strt:splitmean} and~\ref{strt:splittail} are identical.
%
%
\section{Advanced Fork-Join Systems and Networks}
\label{sec:advanced}
Next, we extend the methods from Sec.~\ref{sec:forkjoin} to include $(k,l)$ fork-join systems as well as multi-stage fork-join networks.
%
%
\subsection{$(k,l)$ Fork-Join Systems}
\label{sec:klforkjoin}
An important generalization of fork-join systems are $(k,l)$ fork-join systems, where a job is completed once $l$ out of the $k$ tasks of a job are finished. An application is the download of coded media from distributed storage systems~\cite{joshi:knforkjoin}. Compared to~\cite{joshi:knforkjoin}, we do not assume that the remaining $k-l$ unfinished tasks of a job are dropped once $l$ tasks of the job are finished. This applies, e.g., in case of the transmission of encoded redundant data streams over multiple paths. \cond{long}{Another practical example are systems with redundant jobs that correspond to $(k,1)$ fork-join systems.}

The advantage of $(k,l)$ fork-join systems becomes evident if the service times at each of the $k$ parallel servers are statistically independent. The result of Th.~\ref{th:gg1} does, however, not use an assumption of independence as the individual waiting and sojourn times of the parallel servers are stochastically dependent on the common arrival process~\cite{baccelli:forkjoin, kemper:forkjoin}. Hence, in order to take advantage of independent service times $S_i(n)$ for servers $i \in [1,k]$, we require a service characterization for each of the servers that can be composed using the independence assumption. We derive a suitable model using statistical envelope functions of the type of~\cite{yin:generalizedstochasticallyboundedburstiness}.
\begin{definition}
\label{def:envelopes}
An arrival process $A(m,n)$ has a statistical sample path envelope $\rho_A$ with error profile $\varepsilon_A(\tau_A)$ for $\tau_A \ge 0$ if it holds for all $n \ge 1$ that
\begin{equation*}
\mathsf{P} \left[ \max_{\nu \in [1,n]} \{\rho_A (n-\nu) - A(\nu,n)\} > \tau_A \right] \le \varepsilon_A(\tau_A) .
\end{equation*}
Similarly, a service process $S(m,n)$ has envelope $\rho_S$ with error profile $\varepsilon_S(\tau_S)$ for $\tau_S \ge 0$ if it holds for all $n \ge 1$ that
\begin{equation*}
\mathsf{P} \left[ \max_{\nu \in [1,n]} \{S(\nu,n) - \rho_S (n-\nu+1)\} > \tau_S \right] \le \varepsilon_S(\tau_S) .
\end{equation*}
\end{definition}
\begin{corollary}
\label{cor:errorprofiles}
Given iid inter-arrival and service times with parameters $\rho_A(-\theta_A)$~\eqref{eq:arrivalparameter} and $\rho_S(\theta_S)$~\eqref{eq:serviceparameter} for $\theta_A, \theta_S > 0$. It follows that $\rho_A(-\theta_A)$ and $\rho_S(\theta_S)$ satisfy Def.~\ref{def:envelopes} with error profile $\varepsilon_A(\tau_A) = e^{-\theta_A \tau_A}$ and $\varepsilon_S(\tau_S) = e^{-\theta_S \tau_S}$, respectively.
\end{corollary}
\cond{short}{The proof is a variation of the GI$\mid$GI$\mid$1 case of Th.~\ref{th:gg1}. We note that a similar envelope characterization that differs mainly by a pre-factor of the error profile can also be obtained for the G$\mid$G$\mid$1 case. We omit the result for reasons of space.}\cond{long}{The proof is a variation of the GI$\mid$GI$\mid$1 case of Th.~\ref{th:gg1}. It is given in the appendix. We note that a similar envelope characterization that differs mainly by a pre-factor of the error profile can also be obtained for the G$\mid$G$\mid$1 case, see~\cite{jiang:maxplus} for a general derivation. We omit the result for reasons of space.}

The important quality of the envelopes specified in Def.~\ref{def:envelopes} is the consideration of sample paths of the arrivals and the service by use of the $\max$ operator.
\cond{long}{As a result, the service envelope can be reformulated as
\begin{equation*}
\mathsf{P} \left[ \forall \nu \in [1,n] : S(\nu,n) \le \rho_S(n-\nu+1) + \tau_S \right] \ge 1-\varepsilon_S(\tau_S).
\end{equation*}
By substitution of the envelope for $S(\nu,n)$ into~\eqref{eq:maxplusserviceprocess} we obtain
\begin{multline*}
\mathsf{P} \left[  D(n) \le \max_{\nu \in [1,n]} \{ A(\nu) + \rho_{S} (n-\nu+1) \} + \tau_{S} \right] \\ \ge 1-\varepsilon_S(\tau_{S}) ,
\end{multline*}
so that finally
\begin{multline}
\mathsf{P} \left[ D(n) > \max_{\nu \in [1,n]} \{ A(\nu) + \rho_{S} (n-\nu+1) \} + \tau_{S} \right] \\ \le \varepsilon_S(\tau_{S}) ,
\label{eq:statisticalservicecurve}
\end{multline}
for $n \ge 1$.}\cond{short}{As a result, the service envelope can be substituted for $S(\nu,n)$ into~\eqref{eq:exactmaxplusserviceprocess} to obtain
\begin{equation}
\mathsf{P} \left[ D(n) > \max_{\nu \in [1,n]} \{ A(\nu) + \rho_{S} (n-\nu+1) \} + \tau_{S} \right] \le \varepsilon_S(\tau_{S}) ,
\label{eq:statisticalservicecurve}
\end{equation}
for $n \ge 1$.}
\cond{long}{By insertion of~\eqref{eq:statisticalservicecurve} into $T(n) = D(n)-A(n)$ a statistical bound of the sojourn time follows for $n \ge 1$ as
\begin{equation*}
\mathsf{P}\!\left[T(n) \! > \!\! \max_{\nu \in [1,n]} \{\rho_S (n-\nu+1) - A(\nu,n) \} + \tau_S \right] \le \varepsilon_S(\tau_S) .
\end{equation*}
A similar substitution of the arrival envelope from Def.~\ref{def:envelopes} for $A(\nu,n)$ and using the union bound yields
\begin{multline*}
\!\!\mathsf{P}\!\left[T(n) \!>\! \max_{\nu \in [1,n]} \{\rho_S (n-\nu+1) - \rho_A(n-\nu) \} + \tau_A + \tau_S \right] \\ \le \varepsilon_A(\tau_A) + \varepsilon_S(\tau_S) ,
\end{multline*}
so that for $\rho_S \le \rho_A$ we have
\begin{equation}
\mathsf{P}[T(n) > \tau_A + \tau_S + \rho_S] \le \varepsilon_A(\tau_A) + \varepsilon_S(\tau_S) .
\label{eq:sojourntimeenvelopes}
\end{equation}}
\cond{short}{By insertion of~\eqref{eq:statisticalservicecurve} into $T(n) = D(n)-A(n)$ and substitution of the arrival envelope from Def.~\ref{def:envelopes} for $A(\nu,n)$ a statistical bound of the sojourn time follows for $n \ge 1$ as
\begin{equation}
\mathsf{P}[T(n) > \tau_A + \tau_S + \rho_S] \le \varepsilon_A(\tau_A) + \varepsilon_S(\tau_S),
\label{eq:sojourntimeenvelopes}
\end{equation}
for $\rho_S \le \rho_A$. }Finally, Cor.~\ref{cor:errorprofiles} gives $\varepsilon_A(\tau_A) = e^{-\theta_A \tau_A}$ and $\varepsilon_S(\tau_S) = e^{-\theta_S \tau_S}$ for any $\theta_A,\theta_S > 0$ that satisfy $\rho_S(\theta_S) \le \rho_A(-\theta_A)$. We note that~\eqref{eq:sojourntimeenvelopes} does not assume independence of the arrivals and the service.

Considering $(k,l)$ fork-join systems, job $n \ge 1$ departs once $l$ out of its $k$ tasks are finished. The departure time follows as
\begin{equation}
D(n) = \min_{C \in \mathbb{C}(k,l)} \left \{ \max_{i \in C} \{ D_i(n) \} \right\},
\label{eq:lkforkjoindepartures}
\end{equation}
where $\mathbb{C}(k,l)$ is the set of all combinations of size $l$ that can be chosen from $k$. For the special case of a standard fork-join system we have $l=k$ so that $\mathbb{C}$ comprises one combination that contains all $i \in [1,k]$ so that~\eqref{eq:lkforkjoindepartures} becomes $D(n) = \max_{i \in [1,k]} \{ D_i(n) \}$. On the other hand, for a system with $k$ redundant jobs we have $l=1$ and $\mathbb{C}$ comprises $k$ different combinations that each contain one element, so that $D(n) = \min_{i \in [1,k]} \{ D_i(n) \}$. In general,~\eqref{eq:lkforkjoindepartures} considers the $l$th order statistic of $D(n)$.

\begin{figure}
  \centering
  \hspace{-10pt}
  \subfigure[ $k \in \{10,15\}, l=10$ ]{
    \includegraphics[width=0.5\columnwidth]{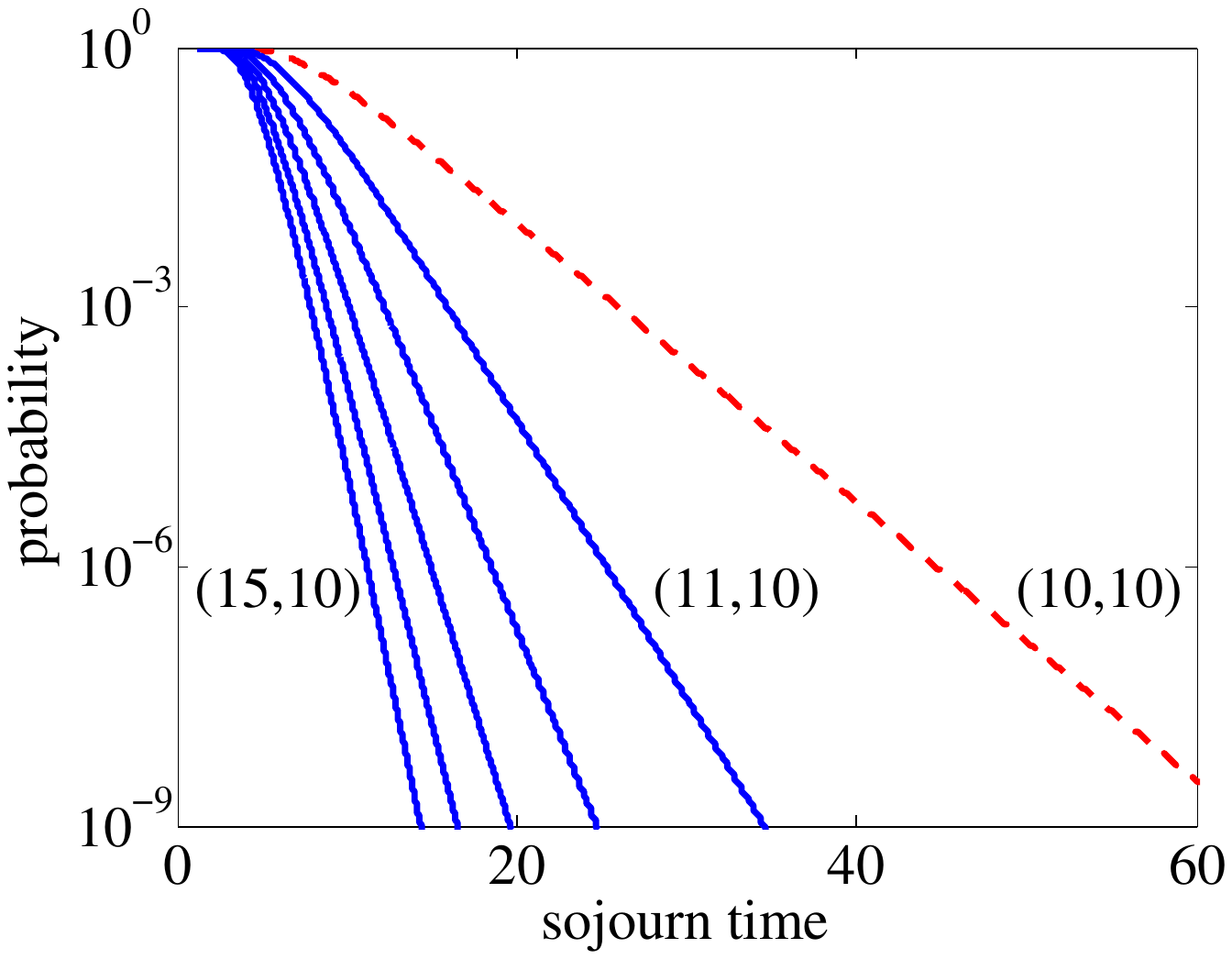}
  }
  \hspace{-15pt}
  \subfigure[ $k \in \{l,l+5\}, \varepsilon = 10^{-6}$ ]{
    \includegraphics[width=0.5\columnwidth]{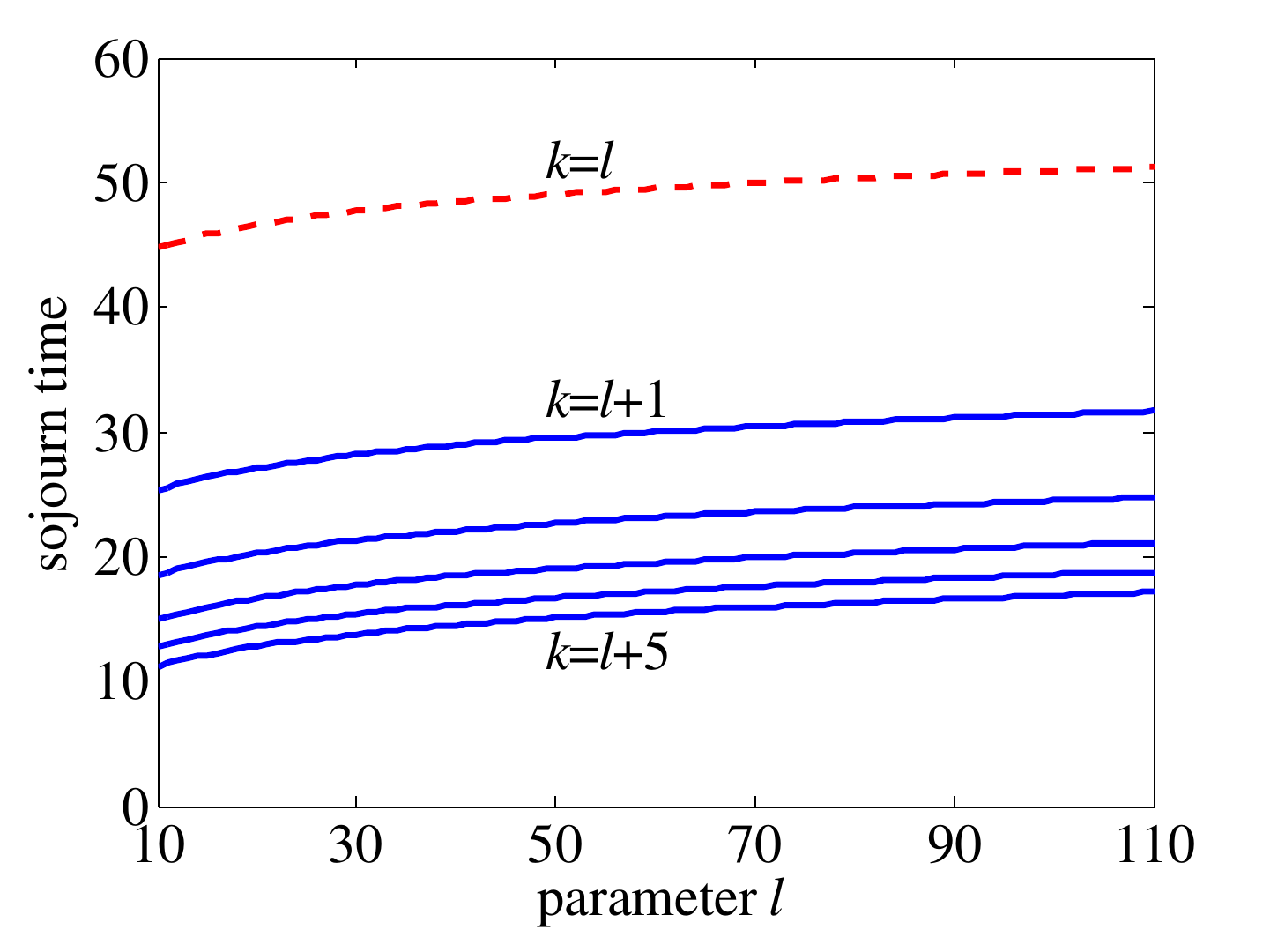}
  }
  \hspace{-10pt}
  \caption{$(k,l)$ fork-join systems achieve a faster tail decay. Increasing redundancy realizes a diminishing improvement.}
  \label{fig:klforkjoin_dm1}
  \cond{short}{\vspace{-10pt}}
\end{figure}
Next, we derive a guarantee of the type of~\eqref{eq:statisticalservicecurve} for $D(n)$ as defined by~\eqref{eq:lkforkjoindepartures}. Given a service envelope for each of the $k$ parallel servers $S_i(\nu,n)$ as in Def.~\ref{def:envelopes}, the departures $D_i(n)$ satisfy~\eqref{eq:statisticalservicecurve} for $n \ge 1$ and $i \in [1,k]$. Since~\eqref{eq:statisticalservicecurve} takes the form $\mathsf{P}[D_i(n) > x] \le p_i$ and conversely $\mathsf{P}[D_i(n) \le x] \ge 1- p_i$,~\eqref{eq:statisticalservicecurve} can be viewed as independent Bernoulli random variables, assuming independence of $S_i(\nu,n)$ for $i \in [1,k]$. In the homogeneous case we have $\rho_{S_i} = \rho_S$ and $p_i = p$ for all $i \in [1,k]$. Since the number of successful Bernoulli trials is binomial, we obtain from~\eqref{eq:lkforkjoindepartures} for $n \ge 1$ that
\begin{equation}
\mathsf{P} \left[  D(n) > x \} \right] \le \sum_{j=0}^{l-1} {k \choose j} (1-p)^j p^{k-j} =: \varepsilon_{(k,l)}(\tau_S),
\label{eq:departuresklforkjoin}
\end{equation}
where $x=\max_{\nu \in [1,n]} \{ A(\nu) + \rho_{S} (n-\nu+1)\} + \tau_S$ and $p = e^{-\theta_S \tau_S}$ from Cor.~\ref{cor:errorprofiles}. A statistical bound of the sojourn time follows from~\eqref{eq:sojourntimeenvelopes} by substitution of $\varepsilon_{(k,l)}(\tau_S)$ for $\varepsilon_S(\tau_S)$.
%
%
\paragraph*{D$\mid$M$\mid$1 Servers}
To focus only on the effect of $(k,l)$ fork-join systems, we use deterministic inter-arrival times $\rho_A$, i.e, $A(\nu,n) = \rho_A (n-\nu)$. The service times are iid exponential with parameters $\mu_i = \mu$ for all servers $i \in [1,k]$ so that $\rho_{S}(\theta_S) = \ln(\mu/(\mu-\theta_S))/\theta_S$ for $\theta_S \in (0,\mu)$ from~\eqref{eq:mm1parameters}. We choose the largest $\theta_S$ that satisfies $\rho_S(\theta_S) \le \rho_A$. In Fig.~\ref{fig:klforkjoin_dm1}, we show the tail decay of the sojourn time~\eqref{eq:sojourntimeenvelopes} where the $(k,l)$ fork-join system is governed by~\eqref{eq:departuresklforkjoin}. The parameters are $\rho_A = 1.25$ and $\mu = 1$. Fig.~\ref{fig:klforkjoin_dm1} shows how redundancy, i.e, $k > l$, improves the sojourn time. It becomes apparent that few redundant servers can achieve a significant advantage. Adding more redundancy realizes a diminishing improvement.
%
%
\paragraph*{Latency-rate Servers}
The system model~\eqref{eq:statisticalservicecurve} takes the basic form of a latency-rate function. It can be parameterized by $1/\rho_S$ that has the interpretation of a service rate, and $\tau_S$ that is a latency with error profile $\varepsilon_S(\tau_S)$. We consider $k=2$ parallel servers, each with $\rho_S=1$ and $\varepsilon_S(\tau_S) = e^{-\kappa \tau_S}$ where $\kappa > 0$ is the speed of the tail decay that determines the average latency $1/\kappa$. An example is a packet data stream that can be routed via two different paths. The arrivals have exponential inter-arrival times with parameter $\lambda = 0.7$. In Fig.~\ref{fig:latencyrateforkjoin}, we show bounds of the sojourn time with $\varepsilon = 10^{-6}$, where we use $\kappa$ as a parameter to evaluate three fundamental strategies: 1) use of only one server; 2) use of both servers in parallel with deterministic thinning and resequencing; 3) use of both servers as a $(2,1)$ fork-join system, i.e., with 1-out-of-2 redundancy.

\begin{figure}
\centering
\cond{short}{\includegraphics[width=0.6\columnwidth]{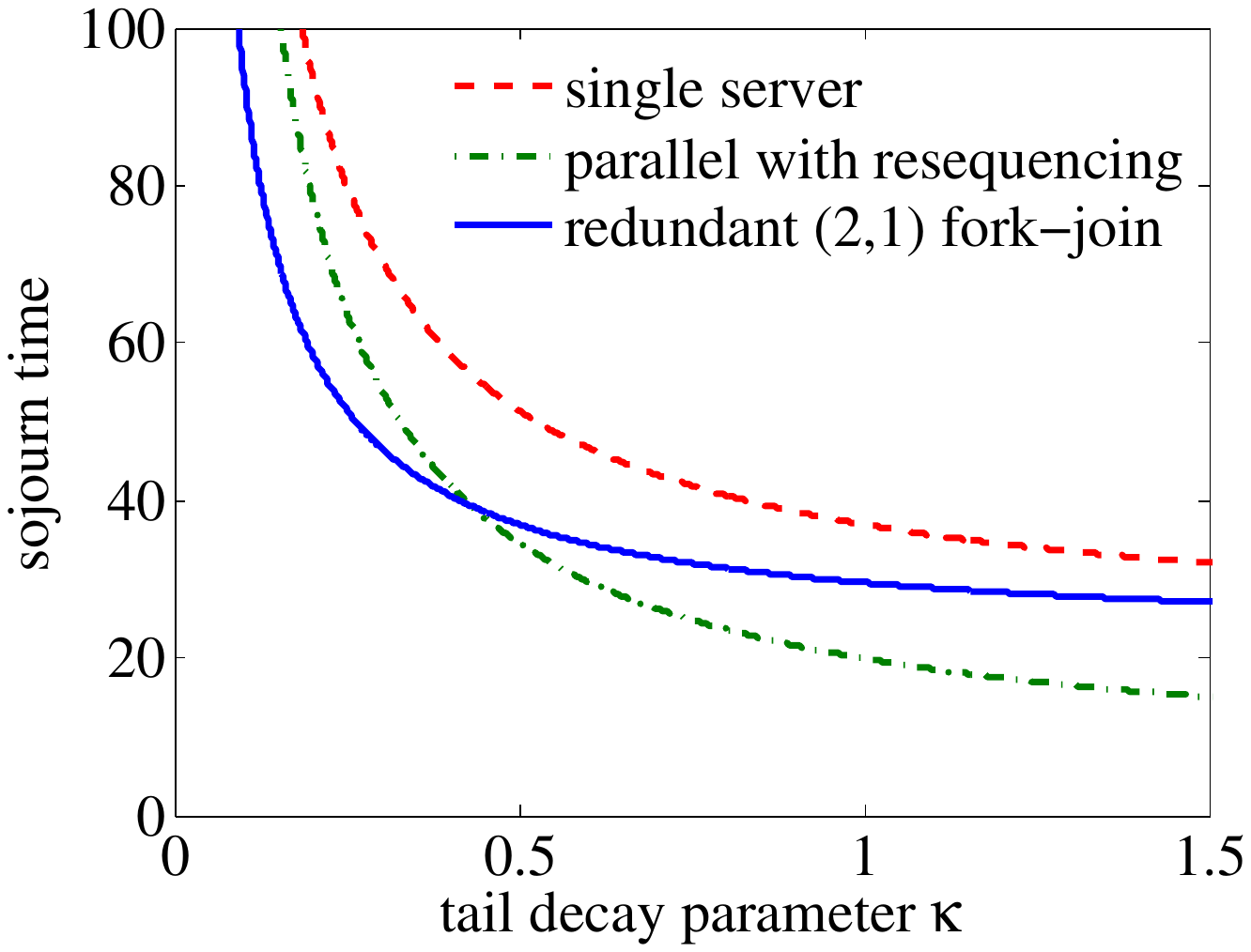}}
\cond{long}{\includegraphics[width=0.75\columnwidth]{latencyrateforkjoin}}
\caption{Two parallel latency-rate servers. The optimal strategy, thinning and resequencing or $(2,1)$ fork-join, depends on the tail decay parameter $\kappa$.}
\cond{short}{\vspace{-10pt}}
\label{fig:latencyrateforkjoin}
\end{figure}
If $\kappa$ is small, the servers frequently experience large delays. Since these delays occur independently for each of the two servers, the redundancy of the $(2,1)$ fork-join system provides a significant advantage. In contrast, if $\kappa$ is large, delays are mostly due to the variability of the arrivals. In case of the $(2,1)$ fork-join system, this affects both servers in the same way, so that the redundant server provides little utility. In this case, the thinning strategy is clearly superior as it benefits from load balancing compared to the two other strategies. Thinning can, however, only combat delays that are due to the variability and the load caused by the arrivals. In contrast, in case of small $\kappa$, thinning and resequencing can result in significant additional delays, as jobs have to wait more frequently at the resequencing stage to ensure in-order delivery. This effect will eventually, if $\kappa$ is reduced further, result in increased sojourn times compared to the single server system.
%
%
\subsection{Multi-Stage Fork-Join Networks}
We derive end-to-end results for multi-stage fork-join networks, where we consider $h$ fork-join stages in tandem, each with $k$ parallel servers. We use subscript $i \in [1,k]$ to distinguish the servers of a stage and superscript $j \in [1,h]$ to denote the stages. We employ the model from Sec.~\ref{sec:klforkjoin} and specify each server by~\eqref{eq:statisticalservicecurve} with parameters $\rho_{S_i}$ and $\varepsilon_{S_i}(\tau_{S})$. For notational convenience, we consider homogeneous servers, i.e., the service processes at each of the servers have the same distribution so that we can drop the subscript $i$. It follows that each stage also satisfies~\eqref{eq:statisticalservicecurve} with parameters $\rho_{stg} = \rho_{S}$ and $\varepsilon_{stg}(\tau_{S})$. In general, we have $\varepsilon_{stg}(\tau_{S}) = k \varepsilon_{S}(\tau_{S})$ from the union bound and for $(k,l)$ fork-join systems with independent parallel servers $\varepsilon_{stg}(\tau_{S}) = \varepsilon_{(k,l)}(\tau_S)$ from~\eqref{eq:departuresklforkjoin}.

In a tandem of fork-join stages, the departures of stage $j$ become the arrivals of stage $j+1$, i.e., $A^{j+1}(n) = D^{j}(n)$ for $j \in [1,h-1]$. As $D^j(n)$ is given by~\eqref{eq:statisticalservicecurve} for all $j \in[1,h]$, it would appear that a solution can be obtained by recursive insertion of~\eqref{eq:statisticalservicecurve} for each stage. The difficulty is, however, that~\eqref{eq:statisticalservicecurve} evaluates sample paths of $A^j(n)$ but does not provide a sample path guarantee for $D^j(n)$, see~\cite{burchard:endtoendstatisticalcalculus} for a discussion of the challenges. To evaluate sample paths of $D^j(n)$ we adapt the basic method from~\cite{ciucu:networkservicecurvescaling2} to max-plus systems and define a sample path guarantee for $D^j(n)$ for any $m \ge 1$ as
\begin{multline}
\mathsf{P} \biggl[ \exists n \in [1,m] : D^j(n) > \max_{\nu \in [1,n]} \{ A^j(\nu) + \rho_{S} (n-\nu+1) \} \\ + \delta (m-n+1) + \tau_S \biggr] \le \varepsilon_{stg}^{\delta}(\tau_S),
\label{eq:singleforkjoinsystempathwise}
\end{multline}
where $\delta > 0$ is a free parameter. Given~\eqref{eq:statisticalservicecurve} holds for stage $j$ with $\varepsilon_{stg} (\tau_S)$ as defined above, we derive \cond{long}{by the union bound that}
\begin{equation}
\varepsilon_{stg}^{\delta} (\tau_S) = \sum_{n=1}^m \varepsilon_{stg} (\tau_S + \delta (m-n+1)) .
\label{eq:samplepathservicecurveerrorprofile}
\end{equation}
\cond{long}{Using~\eqref{eq:singleforkjoinsystempathwise}, a characterization of the end-to-end service of a network of $h$ stages can be derived recursively, starting from the last stage.
For stage $h$, we obtain from~\eqref{eq:singleforkjoinsystempathwise}\footnote{Alternatively, for the special case of stage $h$,~\eqref{eq:statisticalservicecurve} can be used to obtain a tighter guarantee. We use~\eqref{eq:singleforkjoinsystempathwise} here for notational simplicity.} for $n=m$ and $m \ge 1$ that
\begin{multline*}
\mathsf{P} \biggl[ D^h(m) > \max_{\nu \in [1,m]} \{ A^h(\nu) + \rho_{S} (m-\nu+1) \} \\ + \delta + \tau_S \biggr] \le \varepsilon_{stg}^{\delta}(\tau_S).
\end{multline*}
Since $A^h(\nu) = D^{h-1}(\nu)$ and $\nu \in [1,m]$ we can use~\eqref{eq:singleforkjoinsystempathwise} again, to derive
\begin{multline*}
\mathsf{P} \biggl[ D^h(m) > \max_{\nu \in [1,m]} \biggl\{ \max_{\vartheta \in [1,\nu]} \{A^{h-1}(\vartheta) + \rho_S (\nu-\vartheta+1) \} \\ + \delta (m-\nu+1) + \tau_S + \rho_{S} (m-\nu+1) \biggr\} + \delta + \tau_S \biggr] \le 2\varepsilon_{stg}^{\delta}(\tau_S) ,
\end{multline*}
where we used the union bound to estimate the probability. After some reordering we estimate $\delta (m-\nu+1) \le \delta (m-\vartheta+1)$, since $\vartheta \in [1,\nu]$. Finally, we can  combine the outer and the inner $\max$ to arrive at
\begin{multline*}
\mathsf{P} \biggl[ D^h(m) > \max_{\vartheta \in [1,m]} \{A^{h-1}(\vartheta) + \rho_S (m-\vartheta+2) \\ + \delta (m-\vartheta+1) \} + \delta + 2\tau_S \biggr] \le 2\varepsilon_{stg}^{\delta}(\tau_S) .
\end{multline*}
Repeating the same steps, we obtain by recursive insertion of~\eqref{eq:singleforkjoinsystempathwise} for $m \ge 1$ that}\cond{short}{By recursive insertion of~\eqref{eq:singleforkjoinsystempathwise} we obtain for $m \ge 1$ that}
\begin{multline}
\mathsf{P} \biggl[ D^h(m) > \max_{\nu \in [1,m]} \{ A^1(\nu) + \rho_{S} (m-\nu+h) \\ + (h-1)\delta (m-\nu+1) \} + \delta + h \tau_S \biggr] \le h \varepsilon_{stg}^{\delta}(\tau_S) .
\label{eq:networkservicecurve}
\end{multline}
The end-to-end sojourn time of the network follows as before by insertion of~\eqref{eq:networkservicecurve} into $T^{e2e}(m) = D^h(m) - A^1(m)$.
%
%
\paragraph*{GI$\mid$GI$\mid$1 servers}
We derive $\varepsilon_{stg}^{\delta}(\tau_S)$ for a fork-join system with $k$ servers each with iid service times to show how bounds of $T^{e2e}(m)$ scale with $k$ and $h$. With Cor.~\ref{cor:errorprofiles} we have $\varepsilon_{stg}(\tau_S) = k e^{-\theta_S\tau_S}$ so that~\eqref{eq:samplepathservicecurveerrorprofile} can be estimated as
\cond{long}{
\begin{align}
\varepsilon_{stg}^{\delta} (\tau_S) & = \sum_{n=1}^m k e^{-\theta_S(\tau_S+\delta(m-n+1))} \nonumber \\
& \le k e^{-\theta_S\tau_S} \sum_{n=1}^\infty e^{-\theta_S\delta n} \nonumber \nonumber \\
& \le k e^{-\theta_S\tau_S} \int_{0}^{\infty} e^{-\theta_S\delta y} \mathrm{d}y \nonumber \\
& = \frac{k e^{-\theta_S\tau_S}}{\theta_S\delta},
\label{eq:samplepathservicecurveerrorprofilesolution}
\end{align}
where we used that $e^{-\theta_S\delta n} \le \int_{n-1}^n e^{-\theta_S\delta y} \mathrm{d}y$ for $n \ge 1$ since $e^{-\theta_S\delta y}$ is a decreasing function in $y \ge 0$.}
\cond{short}{
\begin{equation}
\varepsilon_{stg}^{\delta} (\tau_S) \le \int_{0}^{\infty} k e^{-\theta_S(\tau_S + \delta y)} \mathrm{d}y = \frac{k e^{-\theta_S\tau_S}}{\theta_S\delta}.
\label{eq:samplepathservicecurveerrorprofilesolution}
\end{equation}
}
Next, we define a constant $\beta > 0$ where we let $\delta = \beta/h$ and insert~\eqref{eq:samplepathservicecurveerrorprofilesolution} into~\eqref{eq:networkservicecurve}.
\cond{long}{It follows that
\begin{equation*}
h \varepsilon_{stg}^{\delta}(\tau_S) = e^{-\theta_S\left(\tau_S - \frac{1}{\theta_S}\ln\frac{h^2 k}{\theta_S\beta} \right)} = e^{-\theta_S z},
\end{equation*}
where we define $z$ so that $\tau_S = z + \frac{1}{\theta_S}\ln\frac{h^2 k}{\theta_S\beta}$. By insertion of this expression into~\eqref{eq:networkservicecurve} and using another variable substitution, we arrive at}\cond{short}{After some reformulation we obtain}
\begin{multline}
\mathsf{P} \biggl[ D^h(m) > \max_{\nu \in [1,m]} \{ A^1(\nu) + (\rho_{S}(\theta_S)+\beta) (m-\nu) \} + \beta \\ + h \biggl(\rho_S(\theta_S) + \tau_S + \frac{1}{\theta_S} \ln\frac{h^2 k}{\theta_S\beta}\biggr)\biggr] \le e^{-\theta_S \tau_S} .
\label{eq:networkservicecurvescaling}
\end{multline}
\cond{long}{A statistical upper bound of the end-to-end sojourn time $T^{e2e}(m) = D^h(m) - A^1(m)$ follows with~\eqref{eq:networkservicecurvescaling} as
\begin{multline*}
\mathsf{P} \biggl[T^{e2e}(m) > \max_{\nu \in [1,m]} \{ (\rho_{S}(\theta_S)+\beta) (m-\nu) - A^1(\nu,m)\} + \beta \\ + h \biggl(\rho_S(\theta_S) + \tau_S + \frac{1}{\theta_S} \ln\frac{h^2 k}{\theta_S\beta}\biggr)\biggr] \le e^{-\theta_S \tau_S} .
\end{multline*}
Considering arrivals with envelope $\rho_A(-\theta_A)$ as in Def.~\ref{def:envelopes} and Cor.~\ref{cor:errorprofiles}, it follows for $\rho_S(\theta_S)+\beta \le \rho_A(-\theta_A)$ that
\begin{multline*}
\mathsf{P} \biggl[T^{e2e}(m) > \beta + \tau_A + h \biggl(\rho_S(\theta_S) + \tau_S + \frac{1}{\theta_S} \ln\frac{h^2 k}{\theta_S\beta}\biggr)\biggr] \\ \le e^{-\theta_A \tau_A} + e^{-\theta_S \tau_S} .
\end{multline*}
which grows with $\mathcal{O}(h \ln (h k))$.}\cond{short}{Considering arrivals with envelope $\rho_A$ as in Def.~\ref{def:envelopes}, it follows from~\eqref{eq:networkservicecurvescaling} for $\rho_S(\theta_S)\!+\!\beta \!\le\! \rho_A$ that the end-to-end sojourn time $T^{e2e}(m) = D^h(m) - A^1(m)$ has a statistical upper bound that grows with $\mathcal{O}(h \ln (h k))$.
}
The result extends the fundamental scaling $\mathcal{O}(h \ln h)$ for $h$ tandem servers from~\cite{ciucu:networkservicecurvescaling2} to $h$ tandem fork-join stages each with $k$ parallel servers. It holds without assuming independent servers within or across stages.
%
%
\section{Conclusions}
\label{sec:conclusions}
We presented a general model of fork-join systems in max-plus system theory and derived essential solutions for $k$ parallel G$\mid$G$\mid$1 servers, that generalize previous $\ln k$ scaling results. Beyond standard fork-join systems, we considered parallel servers with resequencing, where each server has a thinned arrival process. We included applications to load balancing that showed the significant potential but also revealed fundamental limitations. \cond{long}{Our results showed that deterministic thinning is superior to random thinning.} Taking advantage of an envelope-based approach, we contributed solutions for advanced $(k,l)$ fork-join systems and demonstrated the substantial benefit of using independent redundant servers. We provided insights into the configuration of redundant servers versus load-balancing. We extended the results to multi-stage fork-join networks with $h$ stages, where we discovered a scaling with $\mathcal{O}(h \ln (h k))$.
%
%
\section{Appendix}
\label{sec:appendix}
\begin{proof}[Proof of Th.~\ref{th:gg1}]
We only show the proof of the sojourn time, as the proof of the waiting time follows similarly.
\paragraph*{G$\mid$G$\mid$1 servers}
We obtain from~\eqref{eq:sojourntime} for $\theta > 0$ that
\begin{equation*}
\mathsf{E}[e^{\theta T(n)}] \le \sum_{\nu=1}^n \mathsf{E}[e^{\theta S(\nu,n)}] \mathsf{E}[e^{-\theta A(\nu,n)}] ,
\end{equation*}
where we estimated the maximum by the sum of its arguments and used the statistical independence of arrivals and service. By insertion of the $(\sigma,\rho)$-constraints from Def.~\ref{def:sigmarho} we have
\begin{equation*}
\mathsf{E}[e^{\theta T(n)}] \! \le \! e^{\theta (\sigma_A(-\theta) + \sigma_S(\theta) + \rho_S(\theta))} \! \sum_{\nu=1}^n \! e^{-\theta (\rho_A(-\theta) - \rho_S(\theta)) (n-\nu)} .
\end{equation*}
Next, we estimate
\begin{align*}
\sum_{\nu=1}^n e^{-\theta (\rho_A(-\theta)-\rho_S(\theta)) (n-\nu)} & \le \sum_{\nu=0}^{\infty} \bigl(e^{-\theta (\rho_A(-\theta) - \rho_S(\theta))}\bigr)^{\nu} \\
& = \frac{1}{1-e^{-\theta (\rho_A(-\theta) - \rho_S(\theta))}},
\end{align*}
where we used the geometric sum for $\rho_S(\theta) < \rho_A(-\theta)$. By use of Chernoff's bound $\mathsf{P}[X \ge x] \le e^{-\theta x} \mathsf{E}[e^{\theta X}]$ we obtain
\begin{equation*}
\mathsf{P}[T(n) \ge \tau] \le \frac{e^{\theta (\sigma_A(-\theta) + \sigma_S(\theta))}}{1-e^{-\theta (\rho_A(-\theta) - \rho_S(\theta))}} e^{\theta \rho_S(\theta)} e^{-\theta\tau}.
\end{equation*}
\paragraph*{GI$\mid$GI$\mid$1 servers}
From~\eqref{eq:sojourntime} we have
\begin{equation*}
T(n) = \max_{m \in [1,n]} \{S(n-m+1,n) - A(n-m+1,n) \} .
\end{equation*}
For $\theta > 0$ we can write
\begin{multline*}
\mathsf{P}[T(n) > \tau ] \\ = \mathsf{P}\left[\max_{m \in [1,n]} \left\{e^{\theta (S(n-m+1,n) - A(n-m+1,n))} \right\} > e^{\theta\tau} \right] .
\end{multline*}
Now consider the process
\begin{equation*}
U(m) = e^{\theta (S(n-m+1,n) - A(n-m+1,n))} .
\end{equation*}
By definition of $A(m,n)$~\eqref{eq:arrivalincrements} and $S(m,n)$~\eqref{eq:serviceincrements} it follows that
\begin{equation*}
U(m+1) = U(m) e^{\theta (S(n-m) - A(n-m,n-m+1))} .
\end{equation*}
The conditional expectation can be computed as
\begin{align*}
&\mathsf{E}[U(m+1) | U(m), U(m-1),\dots,U(1)] \\
=& U(m) \mathsf{E}[ e^{\theta S(n-m)}] \mathsf{E}[ e^{-\theta  A(n-m,n-m+1)}] ,
\end{align*}
where we used the statistical independence of the inter-arrival and service times. If $\rho_S(\theta) \le \rho_A(-\theta)$, it holds that $\mathsf{E}[ e^{\theta S(n-m)}] \mathsf{E}[ e^{-\theta  A(n-m,n-m+1)}] \le 1$ and
\begin{equation*}
\mathsf{E}[U(m+1) | U(m), U(m-1),\dots,U(1)] \le U(m),
\end{equation*}
i.e., $U(m)$ is a supermartingale. By application of Doob's inequality for submartingales~\cite[Theorem 3.2, p. 314]{doob:stochasticprocesses} and the formulation for supermartingales~\cond{short}{\cite{jiang:onecoin}}\cond{long}{\cite{jiang:onecoin, jiang:noteonsnetcalc}} we have for non-negative $U(m)$ for $m \ge 1$ that
\begin{equation}
x \mathsf{P} \left[ \max_{m \in [1,n]} \{ U(m) \} \ge x \right] \le \mathsf{E}[U(1)] .
\label{eq:martingalebound}
\end{equation}
We derive
\begin{equation*}
\mathsf{E}[U(1)] = \mathsf{E}[e^{\theta (S(n,n) - A(n,n))}] = \mathsf{E}[e^{\theta S(1)}].
\end{equation*}
Letting $x = e^{\theta \tau}$ we have from~\eqref{eq:martingalebound} that
\begin{equation*}
\mathsf{P} \left[ T(n) \ge \tau \right] \le e^{\theta \rho_S(\theta)} e^{-\theta\tau} .
\end{equation*}
Finally, using the union bound to estimate~\eqref{eq:nonblockingsojourntime} gives Th.~\ref{th:gg1}.
\end{proof}
%
%
\cond{long}{
\begin{proof}[Proof of Cor.~\ref{cor:errorprofiles}]
We only show the proof of the arrival envelope as the proof of the service envelope follows similarly.
\paragraph*{GI$\mid$GI$\mid$1 servers}
For $\theta > 0$ we can rewrite Def.~\ref{def:envelopes} as
\begin{align*}
& \mathsf{P}\left[ \max_{\nu \in [1,n]} \{\rho_A(n-\nu) - A(\nu,n) \} > \tau_A \right] \\
= & \mathsf{P}\left[ \max_{m \in [1,n]} \bigl\{e^{\theta (\rho_A(m-1) - A(n-m+1,n))}\bigr\} > e^{\theta\tau_A} \right] .
\end{align*}
Now consider the process
\begin{equation*}
U(m) = e^{\theta (\rho_A(m-1) - A(n-m+1,n))}.
\end{equation*}
By definition of $A(m,n)$~\eqref{eq:arrivalincrements} it follows that
\begin{equation*}
U(m+1) = U(m) e^{\theta (\rho_A - A(n-m,n-m+1))} ,
\end{equation*}
which has conditional expectation
\begin{align*}
& \mathsf{E}[U(m+1)| U(m), U(m-1), \dots, U(1)] \\
= & U(m) e^{\theta \rho_A} \mathsf{E}[e^{-\theta A(n-m,n-m+1)}] .
\end{align*}
By substitution of $\rho_A = \rho_A(-\theta)$ from~\eqref{eq:arrivalparameter} it follows that $\mathsf{E}[U(m+1)| U(m), U(m-1), \dots, U(1)] = U(m)$ is a martingale and hence also a supermartingale. With~\eqref{eq:martingalebound}, $x=e^{\theta \tau_A}$, and since $\mathsf{E}[U(1)]=1$ it holds that
\begin{equation*}
\mathsf{P}\left[ \max_{\nu \in [1,n]} \{\rho_A(n-\nu) - A(\nu,n) \} \ge \tau_A \right] \le e^{-\theta \tau_A} ,
\end{equation*}
which proves Cor.~\ref{cor:errorprofiles}.
\end{proof}
}
%
%
\balance
\bibliographystyle{IEEEtran}
\bibliography{IEEEabrv,ParallelSystems}
\end{document}